\documentclass[11pt,reqno]{amsart}

\pdfoutput = 1

\usepackage[english]{babel}
\usepackage{amsmath,amssymb,amsthm}
\usepackage{amsfonts}
\usepackage{amsxtra}

\usepackage{array,dcolumn}
\usepackage{graphicx}
\usepackage{float}
\usepackage[export]{adjustbox}
\usepackage{subcaption}
\usepackage[hyperfootnotes=true,
        pdffitwindow=true,
        plainpages=false,
        pdfpagelabels=true,
        pdfpagemode=UseOutlines,
        pdfpagelayout=SinglePage,
                hyperindex,]{hyperref}

\usepackage[T1]{fontenc}
\usepackage[utf8]{inputenc}
\usepackage[activate={true,nocompatibility},spacing,kerning]{microtype}
\usepackage{bm}
\usepackage{bbm}
\usepackage[sc,osf]{mathpazo}
\microtypecontext{spacing=nonfrench}

 \calclayout

\newcommand{\mathsym}[1]{{}}
\newcommand{\unicode}[1]{{}}

\theoremstyle{plain}
\newtheorem{theorem}{Theorem}

\newtheorem{corollary}[theorem]{Corollary}
\newtheorem{proposition}[theorem]{Proposition}

\theoremstyle{definition}

\theoremstyle{remark}
\newtheorem{remark}[theorem]{Remark}

\renewcommand{\theequation}{\arabic{equation}}

\renewcommand{\leq}{\leqslant}

\newcommand{\al}{\alpha}

\newcommand\numberthis{\addtocounter{equation}{1}\tag{\theequation}}
\numberwithin{equation}{section}
\begin{document}
\title[Finite size corrections in the GUE and LUE]{Functional form for the leading correction to the distribution of the largest eigenvalue in the
GUE and LUE}
\author{Peter J. Forrester and Allan K. Trinh}
\address{Department of Mathematics and Statistics, 
ARC Centre of Excellence for Mathematical \& Statistical Frontiers,
University of Melbourne, Victoria 3010, Australia}
\email{pjforr@unimelb.edu.au}
\email{a.trinh4@student.unimelb.edu.au}
\date{\today}


\begin{abstract}
The neighbourhood of the largest eigenvalue $\lambda_{\rm max}$ in the Gaussian unitary ensemble (GUE) and Laguerre unitary ensemble (LUE) is referred to as the soft edge. It is known that there exists
a particular centring and scaling such that the distribution of $\lambda_{\rm max}$ tends to a universal form, with an error term bounded by $1/N^{2/3}$.  We take up the problem of computing the exact functional form of the leading error term in a large $N$ asymptotic expansion for both the GUE and LUE --- two versions of the LUE are considered, one with the parameter $a$ fixed, and the other with $a$ proportional to $N$. 
Both settings in the LUE case allow for an interpretation in terms of the distribution of a particular
weighted path length
in a model involving exponential variables on a rectangular grid, as the grid size gets large.
We give
operator theoretic forms of the corrections, which are corollaries of knowledge of the first two terms in the large $N$ expansion of the scaled kernel, and are readily computed using a method due to Bornemann.
We also give expressions in terms of the solutions of particular systems of coupled differential equations,
which provide an alternative method of computation. Both characterisations are well suited to a thinned generalisation of the original ensemble, whereby each eigenvalue is deleted independently with probability
$(1 - \xi)$. In the final section, we investigate using simulation the question of whether upon an appropriate centring and scaling a wider class of complex Hermitian random matrix ensembles have their
leading correction to the distribution of $\lambda_{\rm max}$ proportional to $1/N^{2/3}$. 
\end{abstract}

\maketitle

\section{Introduction}
In applications of random matrices, one often comes across the statement that the matrices should be
large and scaled. For example, in relation to the Riemann zeros, according to the Montgomery--Odlyzko law
\cite{KS99} the relevant matrices are complex Hermitian of a large size and bulk scaled. In applications to quantum spectra
with time reversal symmetry, it is large sized bulk scaled real symmetric matrices which are relevant; see e.g.~\cite{Bo05}.
For growth models on curved interfaces in the KPZ class, large complex Hermitian matrices of a large size are again relevant, but now with soft edge scaling \cite{PS01}.

As first noticed by Dyson \cite{Dy62}, applications of  random matrices requiring complex Hermitian matrices can often be reformulated to involve instead complex unitary matrices $U \in U(N)$. In the case of the Riemann zeros, this has been shown by Keating and Snaith \cite{KS00a} to have fundamental consequences: the $U(N)$ model permits for the quantitative prediction of certain ``finite size" effects when comparing against Odlyzko's
\cite{Od89} data set (high precision evaluation of the $10^{20}$-th Riemann zero and over 70 million neighbours).
In particular, the effective value of $N$ required in the $U(N)$ model was identified (it is proportional to the
logarithm of the zero number) for purposes of reproducing Odlyzko's plot of the value distribution of the logarithm
of the Riemann zeta function along the critical line in the neighbourhood of the $10^{20}$-th zero.

A data set beginning with the $10^{23}$-rd zero and its $10^9$ neighbours was announced by Odlyzko in 
\cite{Od01}. The greater statistical accuracy inherent in this data set relative to the original allowed for the finite size correction of the deviation of the empirical  nearest neighbour spacing distribution and the limiting random matrix distribution to be displayed. This derivation exhibited clear structure. Bogomolny and collaborators \cite{BBLM06}
took up the task of predicting its functional form. Consistent with the work of Keating and Snaith, a combination of analytic and numerical evidence was presented to exhibit that the $U(N)$ model again correctly accounts for the  derivation.

A problem left open from \cite{BBLM06}, and addressed in \cite{FM15,BFM17} was the analytic specification of the
correction term $r_2(0;s)$ in the large $N$ expansion
\begin{equation}\label{pr}
p^{U(N)}(0;s) = p_2(0;s) + {1 \over N^2} r_2(0;s) + \cdots
\end{equation}
Here $p^{U(N)}(0;s)$ denotes the consecutive eigen-angle spacing distribution for Haar distributed unitary
random matrices, with the angles rescaled to have mean spacing unity. On the RHS the quantity $p_2(0;s)$ ---
where the subscript ``2'' is the beta label from Dyson's three fold way \cite{Dy62c} in the absence
of time reversal symmetry --- is the limiting distribution. Pioneering work in random matrix theory due to Mehta
\cite{Me60} and Gaudin \cite{Ga61} paved the way for Dyson \cite{Dy62a} to obtain the Fredholm
determinant formula
\begin{equation}\label{1.2}
p_2(0;s) = {d^2 \over d s^2} \det (\mathbb I - \mathbb K_s),
\end{equation}
where $\mathbb K_s$ is the integral operator on $90,s)$ with kernel
$$
K(x,y) = {\sin \pi (x-y) \over \pi (x-y)}.
$$
Two decades later the Kyoto school of Jimbo et al.~\cite{JMMS80} put (\ref{1.2}) in the context of the
Painlev\'e theory, obtaining the result
\begin{equation}\label{yr}
 \det (\mathbb I - \xi \mathbb K_s) = \exp \Big (
\int_0^{\pi s} {\sigma^{(0)}(t;\xi) \over t} \, dt \Big ),
\end{equation}
where $\sigma^{(0)}$ satisfies the particular Painlev\'e V equation in sigma form
(see e.g.~\cite[Ch.~8]{Fo10} in relation to this class of nonlinear differential equation)
$$
(t \sigma'')^2 + 4 (t \sigma' - \sigma + (\sigma')^2) = 0,
$$
with small $t$ boundary condition
$$
\sigma^{(0)}(t;\xi) = - {\xi \over \pi} t - {\xi^2 \over \pi^2} t^2 +
{\rm O}(t^3).
$$
An expression of the type (\ref{yr}) is referred to as a $\tau$-function, in the sense of the Kyoto school.

It was shown in \cite{FM15} that analogous to $p_2(0;s)$, $r_2(0;s)$ permits both a Fredholm (operator theoretic) and differential
equation characterisation. With the definition
\begin{equation}\label{OIS}
\Omega(\mathbb K_s) : \mathbb L_S = - \det ( \mathbb I - \mathbb K_s) {\rm Tr} \, ((\mathbb I - \mathbb K_s)^{-1}
\mathbb L_s ),
\end{equation}
the former reads
$$
r_2(0;s) = {d^2 \over d s^2} \Omega(\mathbb K_s) : \mathbb L_S,
$$
while the latter involves a second order linear homogenous equation with coefficients given in terms of $\sigma^{(0)}$ .

For the $U(N)$ ensemble, the task of determining the functional form of a number of other spacing--type distributions was also undertaken in \cite{FM15,BFM17}, again motivated by a statistical analysis of Odlyko's data set. As
remarked in the opening paragraph, there are a number of applications in random matrix theory relating to distributions in both
the bulk and edge scaling regimes; \S \ref{s4.4} details one example of the latter involving the length
of a weighted path in a particular directed growth model.
 In the soft edge scaling regime, the specification of the analogue of $r_2(0;s)$ in (\ref{pr}) cannot be found in the earlier random matrix literature, although it is known that upon appropriate scaling
 and centring, the correction term for a wide class of classical random matrix
 ensemble is proportional to $1/N^{2/3}$ \cite{EK06,Jo08,JM12,Ma12}.
 We obtain characterisations of the precise functional form of this leading correction term for the GUE
 and LUE, and for the LUE consider both the case of the parameter $a$ fixed independent of $N$, and when $a$ is proportional to $N$. One characterisation is operator theoretic, and the other involves coupled differential equations. Both are suited to numerical computation, and allow too for a generalisation
 whereby each eigenvalue is deleted independently with probability $(1 - \xi)$. 
In the final section, an investigation using simulation of the distribution of the largest eigenvalue of a non-classical ensemble after subtraction of the leading form is made from the viewpoint of observing a $1/N^{2/3}$
correction.

\section{Operator theoretic formulae}\label{s2}
\subsection{Expansion of Fredholm determinants}

Up to normalisation, the eigenvalue PDF of the GUE and LUE have the functional form (see e.g. \cite[5.4.1]{Fo10})
\begin{equation}\label{2.1}
\prod_{l=1}^Nw(x_l)\prod_{1\leq j<k\leq N}(x_k-x_j)^2,
\end{equation}
where, with $\chi_T=1$ if $T$ is true, $\chi_T=0$ otherwise,
\begin{equation}\label{2.2}
w(x)=\begin{cases}
e^{-x^2},\quad &\text{GUE}
\\ x^ae^{-x}\chi_{x>0},\quad &\text{LUE}.
\end{cases}
\end{equation}
Let $\{p_n(x)\}_{n=0,1,...}$ be the set of monic orthogonal polynomials associated with $w(x)$,
\begin{equation}\label{2.3}
\int_{-\infty}^\infty w(x)p_m(x)p_n(x)\,\mathrm{d}x=h_n\delta_{m,n}.
\end{equation} 
With $H_n(x), L_n^a(x)$ denoting the Hermite and Laguerre polynomials, one has
\begin{equation}\label{2.4}
p_n(x)=\begin{cases}
2^{-n}H_n(x)
\\ (-1)^n n!L_n^a(x)
\end{cases},\quad
h_n=\begin{cases}
\pi^{1/2}2^{-n}n!
\\ n!\Gamma(a+n+1)
\end{cases}
\end{equation}
for the GUE and LUE respectively.

A fundamental property of the PDF \eqref{2.1} is that it specifies a determinantal point process. As such the general $k$-point correlation function is determined by a so-called kernel function $K_N(x,y)$ of just two variables according to
\begin{equation}\label{2.5}
\rho_{(k)}(x_1,...,x_k)=\det[K_N(x_i,x_j)]_{i,j=1,...,k}.
\end{equation}
Moreover, the kernel function is given in terms of the orthogonal polynomials by
\begin{align}\label{2.6}
K_N(x,y)&=\left(w(x)w(y)\right)^{1/2}\sum_{n=0}^{N-1}\frac{p_n(x)p_n(y)}{h_n} \nonumber
\\ &=\frac{\left(w(x)w(y)\right)^{1/2}}{h_{N-1}}\frac{p_N(x)p_{N-1}(y)-p_{N-1}(x)p_N(y)}{x-y},
\end{align}
with the equality in the second line known as the Christoffel-Darboux summation formula.

Let $J$ be a domain within the support of \eqref{2.1} and denote by $E_N(0;J)$ the probability that $J$ contains no eigenvalues. A fundamental corollary of the determinantal formula \eqref{2.5} is the operator theoretic determinant formula (see e.g. \cite[Prop. 5.2.1 and 5.2.2]{Fo10})
\begin{equation}\label{2.7}
E_N(0;J)=\det(\mathbb{I}-\mathbb{K}_{N,J}),
\end{equation}
where $\mathbb{I}$ is the identity operator and $\mathbb{K}_{N,J}$ is the integral operator supported on $J$ with kernel \eqref{2.6}.

We specialise now to the semi-infinite interval $J=(s,\infty)$, so that $E_N(0;J)$ is the probability that all eigenvalues are less than or equal to $s$. For a soft edge scaling, we want to furthermore choose the origin of the scaling variable, $t$ say, so that $t=0$ corresponds to the leading order boundary of the support of the eigenvalues, and that spacing between eigenvalues at the edge is of order unity in the variable $t$. This is achieved by the change of variables \cite{Fo93a}
\begin{equation}\label{2.8}
s=s_t=\begin{cases}
\sqrt{2N}+t/\sqrt{2}N^{1/6}
\\ 4N+2a+2(2N)^{1/3}t
\end{cases}
\end{equation}
for the GUE and LUE respectively. The PDF of the largest soft edge scaled eigenvalue is then
\begin{equation}\label{2.9}
p_N(t)=\frac{\mathrm{d}}{\mathrm{d}t}\det(\mathbb{I}-\mathbb{K}_{N,(s_t,\infty)}).
\end{equation}

As an extension of the above setting, suppose each eigenvalue is deleted independently with probability $(1-\xi)$, where $0<\xi\leq 1$. This process is referred to as thinning. It has attracted a lot of recent attention in the random matrix theory literature \cite{Fo14a,FM15,BFM17,CC17,BD17,La16,BB17}, after having been introduced over a decade earlier by Bohigas and Pato \cite{BP04}. With $p_N^\xi(t)$ denoting the PDF of the corresponding largest eigenvalue, the formula \eqref{2.9} has the simple generalisation
\begin{equation}\label{2.10}
p_N^\xi(t)=\frac{\mathrm{d}}{\mathrm{d}t}\det(\mathbb{I}-\xi\mathbb{K}_{N,(s_t,\infty)}).
\end{equation}

For both the GUE and LUE (with $a$ fixed independent of $N$), the first two terms in the large $N$ expansion of the soft edge scaled correlation kernel are known explicitly, as is a bound on the remainder \cite{Ch06}. We state this result next, then proceed to show how it can be used to expand \eqref{2.9} for large $N$ up to and including the first correction.
\begin{proposition}(Choup \cite{Ch06})\label{p1}
Consider the GUE and LUE with $a$ fixed. Specify the correlation kernel $K_N(x,y)$ by \eqref{2.6} with quantities there given by \eqref{2.2} and \eqref{2.4}. Replace $x$ and $y$ by the soft edge scaling variables $s_x$ and $s_y$ \eqref{2.8}.

For large $N$ one has
\begin{equation}\label{2.11}
\left(\frac{\partial s_x}{\partial x}\right)K_N(s_x,s_y)=K(x,y)+\frac{1}{N^{2/3}}L(x,y)+{\rm O}\left(\frac{1}{N}\right){\rm O}(e^{-x-y})
\end{equation}
where
\begin{equation}\label{2.12}
K(x,y)=\frac{\mathrm{Ai}(x)\mathrm{Ai}'(y)-\mathrm{Ai}'(x)\mathrm{Ai}(y)}{x-y}
\end{equation}
and
\begin{equation*}
L(x,y)=
\begin{cases}
\begin{split}
\frac{1}{20}\bigg[(x+y)\mathrm{Ai}'(x)\mathrm{Ai}'(y)&-(x^2+xy+y^2)\mathrm{Ai}(x)\mathrm{Ai}(y)
\\ &+\frac{3}{2}(\mathrm{Ai}'(x)\mathrm{Ai}(y)+\mathrm{Ai}(x)\mathrm{Ai}'(y))\bigg],
\end{split}\quad &\text{GUE},
\\ 
\begin{split}
\frac{2^{1/3}}{10}\bigg[(x^2+xy+y^2)\mathrm{Ai}'(x)\mathrm{Ai}'(y)&-(x^2+xy+y^2)\mathrm{Ai}(x)\mathrm{Ai}(y) \nonumber
\\ &+\frac{3}{2}(\mathrm{Ai}'(x)\mathrm{Ai}(y)+\mathrm{Ai}(x)\mathrm{Ai}'(y))\bigg],
\end{split}\quad &\text{LUE}.\numberthis\label{2.13}
\end{cases}
\end{equation*}
\end{proposition}

\begin{corollary}\label{cor2}
Let $\mathbb{K}^N_{(s_t,\infty)}$ be the integral operator on $(s_t,\infty)$ with kernel \eqref{2.6} in the Gaussian or Laguerre case, $s_t$ denoting the soft edge scaled variable \eqref{2.8}. Let $\mathbb{K}_{(t,\infty)}$ denote the integral operator on $(t,\infty)$ with kernel \eqref{2.12}, and $\mathbb{L}_{(t,\infty)}$ denote the integral operator on $(t,\infty)$ with kernel \eqref{2.13}. Define $\Omega(\xi\mathbb{K}_{(t,\infty)}): \xi \mathbb{L}_{(t,\infty)}$ according to 
 (\ref{OIS}). We have
\begin{equation}\label{2.14}
\det(\mathbb{1}-\xi\mathbb{K}^N_{(s_t,\infty)})=\det(\mathbb{1}-\xi\mathbb{K}_{(t,\infty)})+\frac{1}{N^{2/3}}\Omega(\xi\mathbb{K}_{(t,\infty)}):\xi \mathbb{L}_{(t,\infty)}+{\rm O}\left(\frac{1}{N}\right)
\end{equation}
and thus
\begin{equation}\label{2.15}
p_N^\xi(t)=p_{0,\infty}^\xi(t)+\frac{1}{N^{2/3}}p_{1,\infty}^\xi(t)+{\rm O}\left(\frac{1}{N}\right)
\end{equation}
with
\begin{equation}\label{2.16}
p_{0,\infty}^\xi(t)=\frac{\mathrm{d}}{\mathrm{d}t}\det(\mathbb{1}-\xi\mathbb{K}_{(t,\infty)})
\end{equation}
and
\begin{equation}\label{2.17}
p_{1,\infty}^\xi(t)=\frac{\mathrm{d}}{\mathrm{d}t}\Omega(\xi\mathbb{K}_{(t,\infty)}):\xi \mathbb{L}_{(t,\infty)}.
\end{equation}
\end{corollary}
\begin{proof}
In \cite{BFM17} an analogous result was established except in a setting that the interval $J$ was bounded. In the present setting, after scaling $J$ is the semi-infinite interval $(t,\infty)$. But the fast decay of $K(x,y)$, $L(x,y)$ and of the remainder in \eqref{2.11} compensate for this, allowing the reasoning from \cite{BFM17} to again be applied to deduce \eqref{2.14}.
\end{proof}

\subsection{Numerical evaluations}\label{NE}

Bornemann \cite{Bo08,Bo09} has shown how to carry out the numerical evaluation of $\det(\mathbb{I}-\mathbb{K}_J)$ with exponentially fast convergence in all cases that $K(x,y)$ is analytic in a neighbourhood of $J$. Extension of these ideas to allow for the computation of $\Omega(\mathbb{K}):\mathbb{L}$ was given subsequently in \cite{BFM17}.

The method relies on a quadrature rule
\begin{equation*}
\int_Jf(x)\,\mathrm{d}x\approx\sum_{j=1}^nf(x_j)w_j
\end{equation*}
with positive weights $w_j$. Data from the quadrature rule is used to construct the Nystr\"om matrix
\begin{equation*}
\mathbb{K}_w=\left(K(x_j,x_k)w_k\right)_{j,k=1}^n.
\end{equation*}
This provides an approximation to the original Fredholm determinant
\begin{equation*}
\det(\mathbb{I}- \xi \mathbb{K})\approx\det(\mathbb{I}- \xi \mathbb{K}_w)
\end{equation*}
with the sought fast convergence properties. In the case of $\xi \Omega(\mathbb{K}): \xi \mathbb{L}$, the Nystr\"om matrices for both $\mathbb{K}$ and $\mathbb{L}$ are required, and the approximation \cite{BFM17}
\begin{equation*}
\Omega(\xi \mathbb{K}):\xi \mathbb{L}\approx-\det(\mathbb{I}- \xi \mathbb{K}_w) {\rm Tr}  \left((\mathbb{I}-\xi \mathbb{K}_w)^{-1}\xi \mathbb{L}_w\right)
\end{equation*}
exhibits fast convergence properties.

For the soft edge scaled GUE and LUE the operator $\mathbb{K}_{(t,\infty)}$ in \eqref{2.16} has the same kernel \eqref{2.12} and thus
\begin{equation*}
p_{0,\infty}^{\xi,\text{GUE}}(t)=p_{0,\infty}^{\xi,\text{LUE}}(t),
\end{equation*}
which is an example of the soft edge universality. However the kernel $L(x,y)$ \eqref{2.13} specifying the operator $\mathbb{L}_{(t,\infty)}$ in \eqref{2.17} is different in the two cases, and so the functional forms $p_{1,\infty}^{\xi,\text{GUE}}(t)$ and $p_{1,\infty}^{\xi,\text{LUE}}(t)$ will be distinct.

In Figures \ref{fig:1} and \ref{fig:2} we have used Bornemann's method outlined above together with the formulas \eqref{2.7}, \eqref{2.6}, \eqref{2.9} -- all using the soft edge scaled variables \eqref{2.8} -- \eqref{2.16} and \eqref{2.12} to compute
\begin{equation}\label{6.1}
N^{2/3}\left(p_N^\xi(s)-p_{0,\infty}^{\xi}(s)\right)
\end{equation}
for the soft edge scaled GUE and LUE with particular $\xi$ (the necessary derivatives are calculated using a central difference approximation). Superimposed are the functional forms $p_{1,\infty}^\xi(s)$ in the respective cases, calculated by applying Bornemann's method to \eqref{2.17}, with kernels corresponding to $\mathbb{K}_{(t,\infty)}$ and $\mathbb{L}_{(t,\infty)}$ given by \eqref{2.12} and \eqref{2.13} respectively. Observe that while the general shapes always agree, in both
the GUE and LUE plots there can systematic deviations between the two curves for certain ranges of $s$ values.
This is to be expected, as (\ref{6.1}) contains all the lower order corrections, not just $p_{1,\infty}^\xi(s)$.

\begin{figure}[H]
	\centering
	\begin{subfigure}[t]{0.5\textwidth}
		\centering
		\includegraphics[height=1.6in]{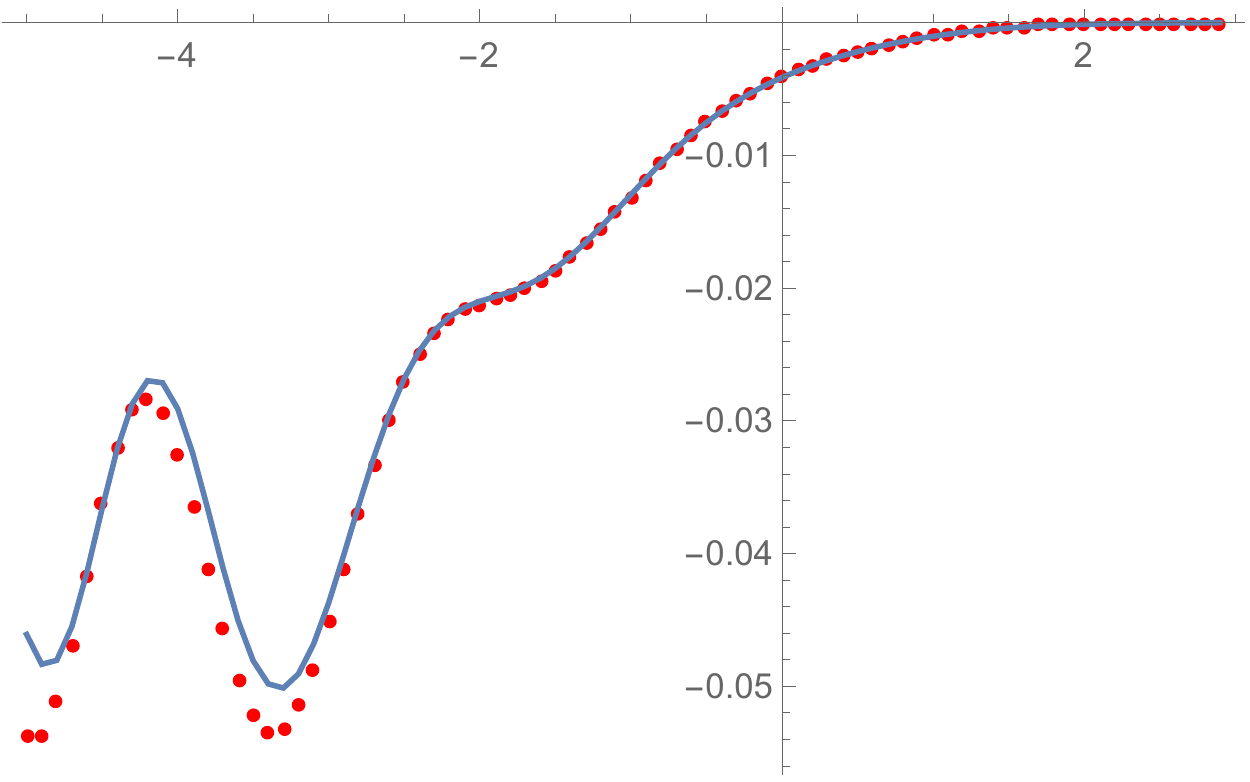}
		\caption{$\xi=0.3$}
	\end{subfigure}%
	~
	\begin{subfigure}[t]{0.5\textwidth}
		\centering
		\includegraphics[height=1.6in]{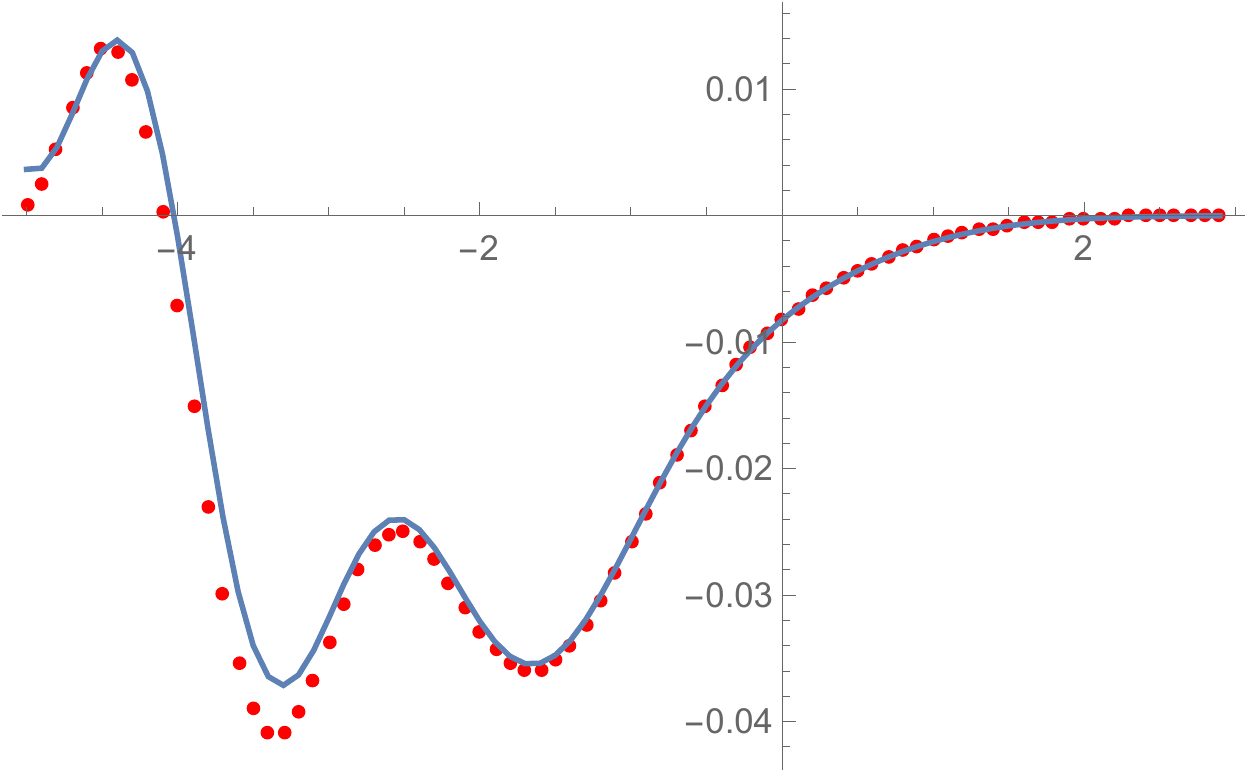}
		\caption{$\xi=0.6$}
	\end{subfigure}
	\caption{GUE. \label{fig:1} Solid line is $p_{1,\infty}^{\xi, G}(s)$, dots are the scaled difference (\ref{6.1}).}
\end{figure}

\begin{figure}[H]
	\centering
	\begin{subfigure}[t]{0.5\textwidth}
		\centering
		\includegraphics[height=1.6in]{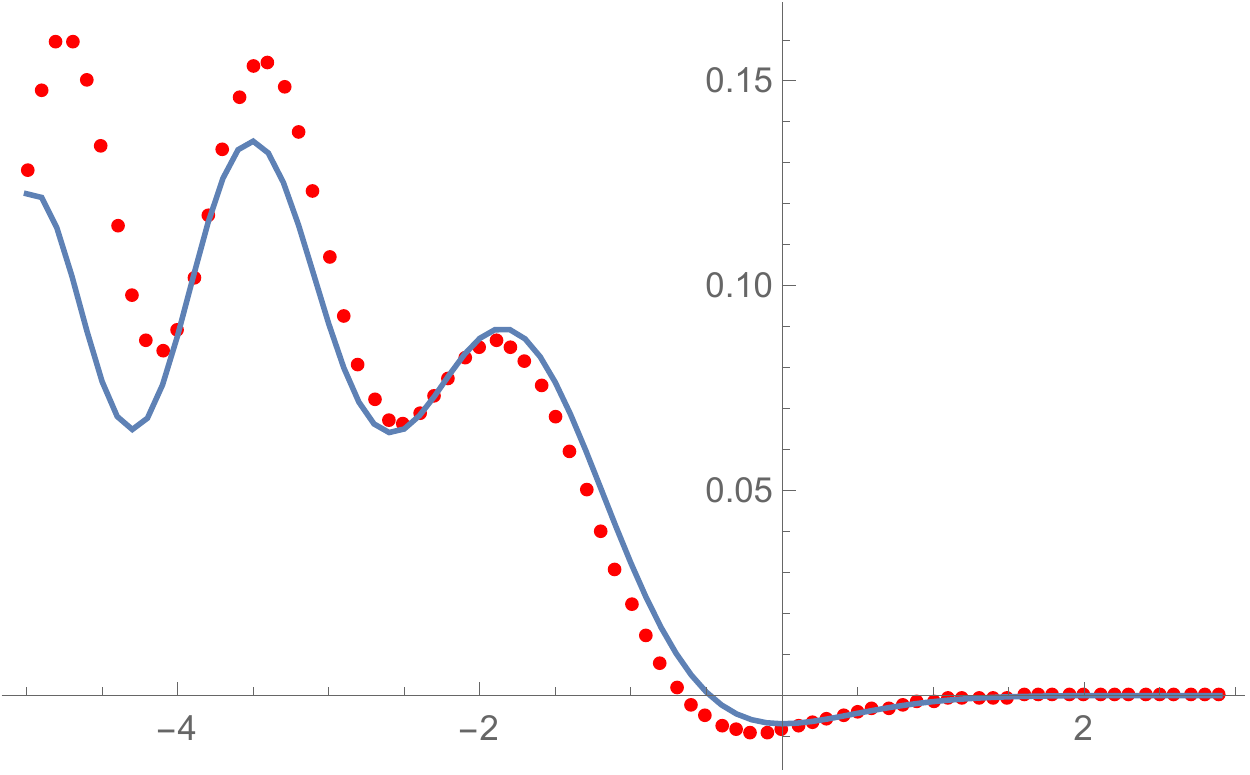}
		\caption{$\xi=0.3$}
	\end{subfigure}%
	~
	\begin{subfigure}[t]{0.5\textwidth}
		\centering
		\includegraphics[height=1.6in]{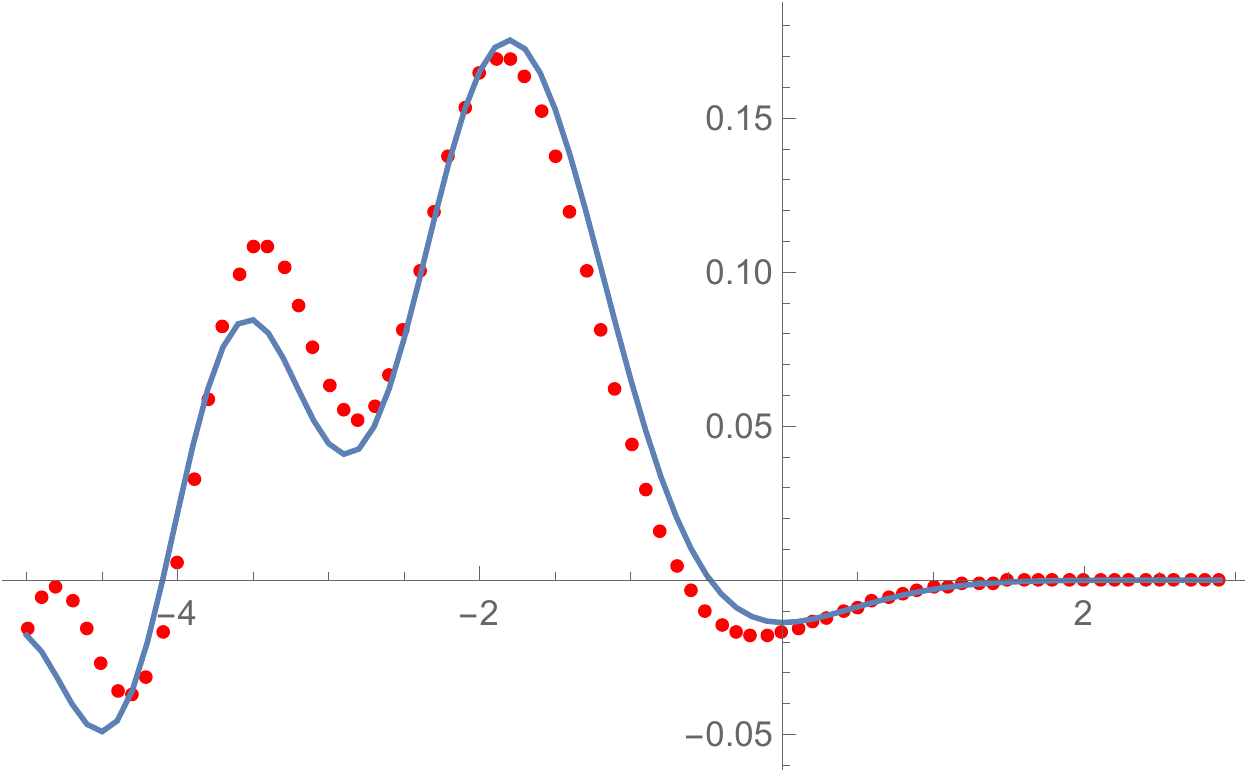}
		\caption{$\xi=0.6$}
	\end{subfigure}
	\caption{LUE with $a=1$. \label{fig:2} Solid line is $p_{1,\infty}^{\xi, L}(s)$, dots are the scaled difference (\ref{6.1}).}
\end{figure}

\section{Characterisation in terms of differential equations}\label{s3}
\subsection{Soft edge GUE}

For the $N\times N$ GUE, in addition to the operator theoretic formula (\ref{2.9}), one has the $\tau$-function expression \cite{TW94c,FW00}
\begin{equation}\label{3.1}
\det\left(\mathbb{I}-\xi\mathbb{K}_{N,(s,\infty)}^{\text{GUE}}\right)=\exp \Big (-\int_s^\infty U_N^G(x;\xi)\,\mathrm{d}x \Big ),
\end{equation}
where $U_N^G$ is the solution of the particular $\sigma$-form of the Painlev\'e IV differential equation
\begin{equation}\label{3.2}
(\sigma'')^2-4(x\sigma'-\sigma)^2+4(\sigma')^2(\sigma'+2N)=0,
\end{equation}
subject to the boundary condition
\begin{equation}\label{3.3}
U_N^G(x;\xi)\underset{x\to\infty}{\sim}\xi K_N^G(x,x).
\end{equation}

Introducing the GUE soft edge change of variables $x=x_y$ as specified by (\ref{2.8}) permits us to expand
\begin{equation}\label{3.4}
\frac{1}{\sqrt{2}N^{1/6}}U_N^G\left(\sqrt{2N}+\frac{y}{\sqrt{2}N^{1/6}}\right)=\sigma_0^G(y)+\frac{1}{N^{2/3}}\sigma_1^G(y)+{\rm O}\left(\frac{1}{N}\right)
\end{equation}
and so obtain a differential equation characterisation of $\sigma_0^G$ and $\sigma_1^G$, which in turn can be used to specify $p_{1,\infty}^{\xi,\text{GUE}}$. This is the same general strategy as that used in \cite{FM15} to obtain a differential equation characterisation of $r_2(0;s)$ in the large $N$ expansion (\ref{pr}).

The boundary conditions in the differential equation characterisation make use of the soft edge scaled expansion of \eqref{3.3}. First, as a point of interest we note from (\ref{2.5}) that $K_N(x,x)=\rho_N^G(x)$, with $\rho_N^G(x)$ denoting the spectral density. For its large $N$ expansion with the soft edge scaling variable $x=x_y$ we can make use of 
(\ref{2.11})--(\ref{2.13}) to deduce (see \cite{GFF05} for a direct computation)
\begin{equation}\label{3.5}
\frac{1}{\sqrt{2}N^{1/6}}\rho_N^G\left(\sqrt{2N}+\frac{y}{\sqrt{2}N^{1/6}}\right)=\rho_{0,\infty}^G(y)+\frac{1}{N^{2/3}}\rho_{1,\infty}^G(y)+{\rm O}\left(\frac{1}{N}\right),
\end{equation}
where
\begin{align*}
\rho_{0,\infty}^G(y)&=\left(\mathrm{Ai}'(y)\right)^2-y\left(\mathrm{Ai}(y)\right)^2
\\ \rho_{1,\infty}^G(y)&=-\frac{1}{20}\left(3y^2\left(\mathrm{Ai}(y)\right)^2-2y\left(\mathrm{Ai}'(y)\right)^2-3\mathrm{Ai}(y)\mathrm{Ai}'(y)\right).
\end{align*}

\begin{proposition}\label{p3}
Consider the expansion (\ref{2.15}) for the GUE. Define $\sigma_0^{\rm G} = \sigma_0^{\rm G}(y;\xi)$ as the solution
of the particular $\sigma$-form Painlev\'e II equation
\begin{equation}\label{3.6}
(\sigma'')^2 + 4 \sigma'((\sigma')^2 - y \sigma' + \sigma) = 0
\end{equation}
subject to the boundary condition
\begin{equation}\label{3.7}
\sigma_0^{\rm G}(y;\xi) \mathop{\sim}\limits_{y \to \infty} \xi \rho_{0,\infty}^{\rm G}(y).
\end{equation}
Define $\sigma_1^{\rm G} = \sigma_1^{\rm G}(y;\xi)$ as the solution of the second order linear
differential equation
\begin{equation}\label{3.8}
A(y) \sigma'' + B(y) \sigma' + C(y) \sigma = D(y),
\end{equation}
where, with $\sigma_0 := \sigma_0^{\rm G}$ as specified above
\begin{align}\label{3.9}
A(y) & = 2 \sigma_0''(y) \nonumber \\
B(y) & = 12  (\sigma_0'(y))^2 - 8 y \sigma_0'(y) + 4 \sigma_0(y) \nonumber \\
C(y) & = 4 \sigma_0'(y) \nonumber \\
D(y) & = (\sigma_0(y))^2 - 2 y \sigma_0(y) \sigma_0'(y) + y^2 ( \sigma_0'(y))^2,
\end{align}
subject to the boundary condition
\begin{equation}\label{3.10}
\sigma_1^{(G)}(y;\xi)\underset{y\to\infty}{\sim}\xi\rho_{1,\infty}^G(y).
\end{equation}
Then, with the symbol $y$ used in place of $t$ as used in \eqref{2.15},
\begin{align*}
p_{0,\infty}^{\xi,G}(y)&=\sigma_0^G(y;\xi)\exp\Big (-\int_y^\infty\sigma_0^G(t;\xi)\,\mathrm{d}t \Big )
\\ p_{1,\infty}^{\xi,G}(y)&=-\sigma_0^G(y;\xi)\left(\int_y^\infty\sigma_1^G(t;\xi)\,\mathrm{d}t\right)
\\
& \quad \times \exp\Big (-\int_y^\infty\sigma_0^G(t;\xi)\,\mathrm{d}t \Big )+\sigma_1^G(y;\xi)\exp\Big (-\int_y^\infty\sigma_0^G(t;\xi)\,\mathrm{d}t \Big ).
\numberthis\label{3.11}
\end{align*}
\end{proposition}
\begin{proof}
As already commented, to obtain the characterisation of $\sigma_0^G$ and $\sigma_1^G$ we introduce the soft edge scaled $y$ into \eqref{3.2} by the change of variables $x=x_y$, then we expand the solution of interest $U_N^G$ according to \eqref{3.4}. The coupled differential equations \eqref{3.6}, \eqref{3.8} then result by equating the first two leading orders in $N$. For the boundary conditions \eqref{3.7} and \eqref{3.10}, we multiply both sides of \eqref{3.3} by $(1/\sqrt{2}N^{1/6})$ then substitute \eqref{3.5} on the RHS and \eqref{3.4} on the LHS. Equating like powers of $N$, \eqref{3.7}, \eqref{3.10} result.

It follows from \eqref{2.9}, \eqref{3.1} and \eqref{3.4} that
\begin{equation}\label{3.11a}
p_N^\xi(y)=\frac{\mathrm{d}}{\mathrm{d}y}\exp\Big (-\int_y^\infty\sigma_0^G(t;\xi)\,\mathrm{d}t-\frac{1}{N^{2/3}}\int_y^\infty\sigma_1^G(t;\xi)\,\mathrm{d}t+{\rm O}\left(\frac{1}{N}\right) \Big ).
\end{equation}
Expanding the RHS to order $1/N^{2/3}$ allows $p_{0,\infty}^{\xi,G}(y)$, $p_{1,\infty}^{\xi,G}(y)$ to be read off in accordance with \eqref{2.15} with the expressions \eqref{3.11} so being obtained.
\end{proof}

\subsection{Soft edge LUE}

The analogue of \eqref{3.1} for the LUE is the $\tau$-function formula \cite{TW94c,FW01a}
\begin{equation}\label{3.12}
\det(\mathbb{I}-\xi\mathbb{K}_{N,(s,\infty)}^\text{LUE})=\exp \Big (-\int_s^\infty U_N^L(x;\xi)\,\frac{\mathrm{d}x}{x}
\Big ),
\end{equation}
where $U_N^L$ is the solution of the particular $\sigma$-form of the Painlev\'e V equation
\begin{equation}\label{3.13}
(x\sigma'')^2-\left(\sigma-x\sigma'+2(\sigma')^2+(a+2N)\sigma'\right)^2+4(\sigma')^2(\sigma'+N)(\sigma'+a+N)=0,
\end{equation}
subject to the boundary condition
\begin{equation}\label{3.14}
\frac{1}{x}U_N^L(x;\xi)\underset{x\to\infty}{\sim}\xi K_N^L(x,x).
\end{equation}

According to \eqref{2.8}, the appropriate soft edge change of variables in \eqref{3.13} is
\begin{equation*}
x=x_y=4N+2a+2(2N)^{1/3}y.
\end{equation*}
The large $N$ expansion of \eqref{3.12} is obtained by expanding
\begin{equation}\label{3.15}
2(2N)^{1/3}\frac{U_N^L(x_y;\xi)}{x_y}=\sigma_0^L(y)+\frac{1}{N^{2/3}}\sigma_1^L(y)+{\rm O}\left(\frac{1}{N}\right).
\end{equation}
Equating leading powers of $N$ in \eqref{3.13} with this change of variables and expansion gives a coupled set of differential equations for $\sigma_0^L$ and $\sigma_1^L$. For the boundary condition, \eqref{2.5} tells us that $K_N^L(x,x)=\rho_N^L(x)$, with  $\rho_N^L(x)$ denoting the spectral density. From the Laguerre case of \eqref{2.11}
\begin{equation}\label{3.16}
2(2N)^{1/3}\rho_N^L(x_y)=\rho_{0,\infty}^L(y)+\frac{1}{N^{2/3}}\rho_{1,\infty}^L(y)+{\rm O}\left(\frac{1}{N}\right),
\end{equation}
where
\begin{align*}
\rho_{0,\infty}^L(y)&=\rho_{0,\infty}^G(y)=\left(\mathrm{Ai}'(y)\right)^2-y\left(\mathrm{Ai}(y)\right)^2
\\ \rho_{1,\infty}^L(y)&= \frac{2^{1/3}}{10}\left(3y^2\left(\mathrm{Ai}(y)\right)^2-2y\left(\mathrm{Ai}'(y)\right)^2+2\mathrm{Ai}(y)\mathrm{Ai}'(y)\right).
\numberthis\label{3.17}
\end{align*}
With knowledge of this expansion we have all necessary information to deduce the Laguerre analogue of Proposition 
\ref{p3}.

\begin{proposition}\label{p4}
Consider the expansion \eqref{2.15} for the LUE with $a$ fixed independent of $N$. Set $\sigma_0^L(y;\xi)=\sigma_0^G(y;\xi)$ where $\sigma_0^G(y;\xi)$ is specified by \eqref{3.6} and \eqref{3.7}. Define $\sigma_1^L$ as the solution of the second order linear differential equation \eqref{3.8} with $A(y)$, $B(y)$, $C(y)$ as in \eqref{3.9} but with $D(y)$ replaced by
\begin{equation}\label{3.18}
D^L(y)=-2^{4/3}\left(2y\sigma_0\sigma_0'-3y^2(\sigma_0')^2+2\sigma_0(\sigma_0')^2+4y(\sigma_0')^3+\sigma_0'\sigma_0''+y(\sigma_0'')^2\right).
\end{equation}
The solution $\sigma_1^L$ is chosen subject to the boundary
\begin{equation}\label{3.19}
\sigma_1^L(y)\underset{y\to\infty}{\sim} \xi \rho_{1,\infty}^L(y).
\end{equation}
The final formulas \eqref{3.11} of Proposition \ref{p3} remain true for the Laguerre case with all superscripts $G$ replaced by $L$. 
\end{proposition}

\subsection{Differential equations for the spectral density}

Substituting \eqref{3.3} in \eqref{3.2} and equating terms of order $\xi^2$ tells us that the spectral density $\rho_N^G$ satisfies the nonlinear differential equation
\begin{equation}\label{3.20}
(\sigma'')^2-4(x\sigma'-\sigma)^2+8N(\sigma')^2=0.
\end{equation}
We note that differentiating with respect to $x$ and cancelling a factor of $2\sigma''$ transforms this to the third order linear differential equation
\begin{equation}\label{3.21}
\sigma'''-4x(x\sigma'-\sigma)+8N\sigma'=0.
\end{equation}
This latter characterisation of $\rho_N^G$ is in fact known from earlier work \cite{GT05,HT12,WF14}.

Analogous reasoning holds for $\rho_N^L$. Thus we substitute \eqref{3.14} in \eqref{3.13} and equate terms of order $\xi^2$ to conclude that the spectral density times $x$, $x\rho_N^L$, satisfies the nonlinear differential equation
\begin{equation}\label{3.22}
(x\mu'')^2-\left(\mu-x\mu'+(a+2N)\mu'\right)^2+4N(a+N)(\mu')^2=0.
\end{equation}
Differentiating with respect to $x$ and cancelling a factor of $2\mu''$ we obtain from this the third order linear differential equation
\begin{equation}\label{3.33}
x^2\mu'''+x\mu''-\left(\mu-x\mu'+(a+2N)\mu'\right)(a+2N-x)+4N(a+N)\mu'=0.
\end{equation}
Writing $\mu=x\sigma$ it follows that $\rho_N^L$ itself satisfies the third order linear differential equation
\begin{equation}\label{3.34}
x^3\sigma'''+4x^2\sigma''-x\left(x^2-2(2N+a)x+a^2-2\right)\sigma'+\left((2N+a)x-a^2\right)\sigma=0.
\end{equation}
Like \eqref{3.21} in relation to $\rho_N^G$, this characterisation of $\rho_N^L$ is in fact known from earlier work \cite{GT05,ATK11,Ra16}.

\subsection{Numerical computations using the coupled differential equations}
In the study \cite{PS03} it was shown how a Painlev\'e transcendent characterisation of $p_{0,\infty}^\xi(y)$ for $\xi = 1$
equivalent to that given in Proposition \ref{p3} could be used to achieve high precision numerical evaluation. In
addition to extending the accuracy of the boundary condition (\ref{3.7}) to the next order, the method made use of
nested power series solutions, with overlapping radii of convergence.

Extending the accuracy of the boundary conditions is straightforward. For example, in the GUE case (\ref{3.7})
is to be extended to read
\begin{equation}\label{3.7a}
\sigma_0^{\rm G}(y;\xi) \mathop{\sim}\limits_{y \to \infty} \xi K(y,y) - 
\xi^2 \int_y^\infty  (K(y,x))^2 \, dx,
\end{equation}
where $K(x,y)$ is given by (\ref{2.12}), and similarly (\ref{3.10}) is to be extended to read
\begin{equation}\label{3.10a}
\sigma_1^{(G)}(y;\xi)\underset{y\to\infty}{\sim}\xi L^{G}(y,y)  - 
\xi^2 \int_y^\infty  (L^{G}(y,x))^2 \, dx.
\end{equation}
These extensions follow from the fact that the extension of the boundary condition (\ref{3.3}) is
\begin{equation}\label{3.3a}
U_N^G(x;\xi)\underset{t\to\infty}{\sim}\xi K_N^G(x,x) - \xi^2
\int_y^\infty (K_N^G(y,x))^2 \, dx,
\end{equation}
which in turn is a corollary of (\ref{3.1}), the expansion (see e.g.~\cite[Eq.~(9.1)]{Fo10})
$$
\det\left(\mathbb{I}-\xi\mathbb{K}_{N,(s,\infty)}^{\text{GUE}}\right) = 1 -
\xi \int_s^\infty \rho_{(1)}(x) \, dx + {\xi^2 \over 2}
 \int_s^\infty  \int_s^\infty  \rho_{(2)}(x_1, x_2) \, dx_1 dx_2 + \cdots
 $$
and the determinantal formula (\ref{2.5}).

With the boundary conditions so extended, we found that using commercial DE solving software gave
agreement, to graphical accuracy at least,
with the operator formulae computed using Bornemann's method for $\xi = 1$ over the
full range of $s$ values considered, but in the cases $\xi < 1$ there was typically a (negative) value of $s$,
occurring at a turning point,
for which the DE solver started tracking the wrong solution;  see Figures  \ref{fig:3} and \ref{fig:4}.

\begin{figure}[H]
	\centering
	\begin{subfigure}[t]{0.5\textwidth}
		\centering
		\includegraphics[height=1.6in]{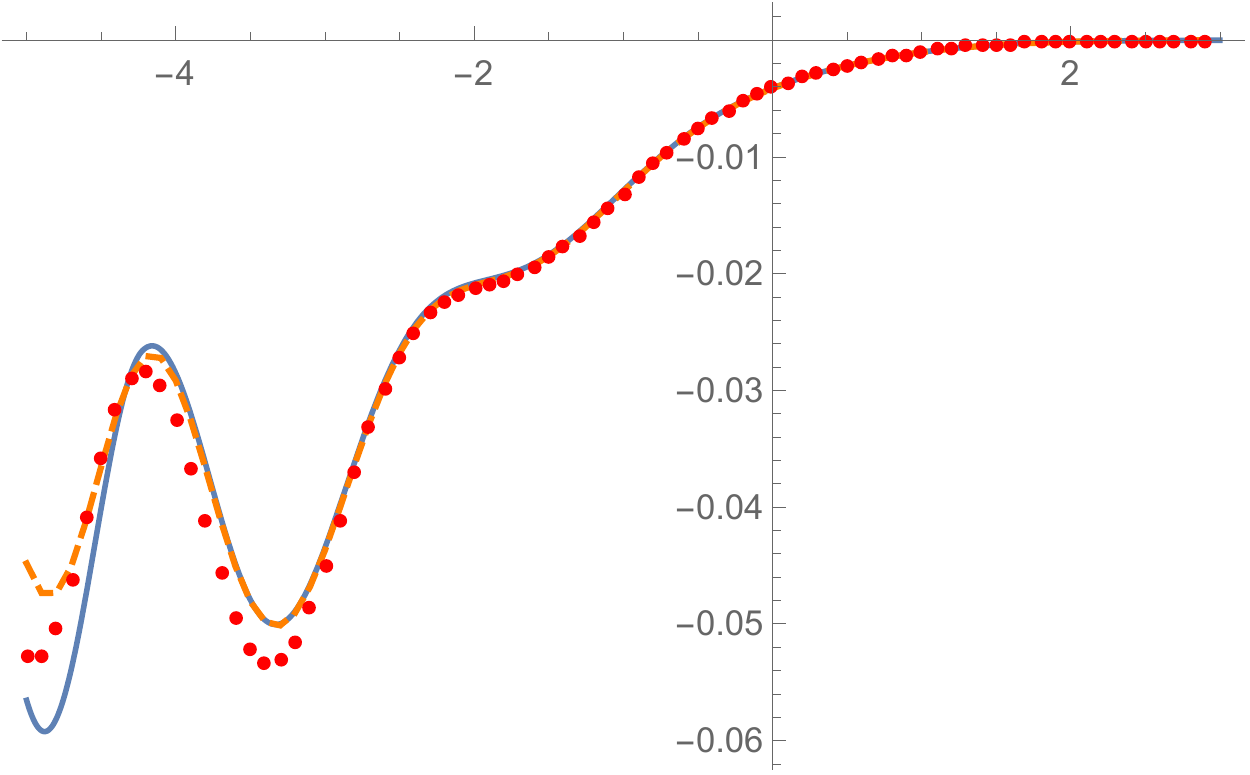}
		\caption{$\xi=0.3$}
	\end{subfigure}%
	~
	\begin{subfigure}[t]{0.5\textwidth}
		\centering
		\includegraphics[height=1.6in]{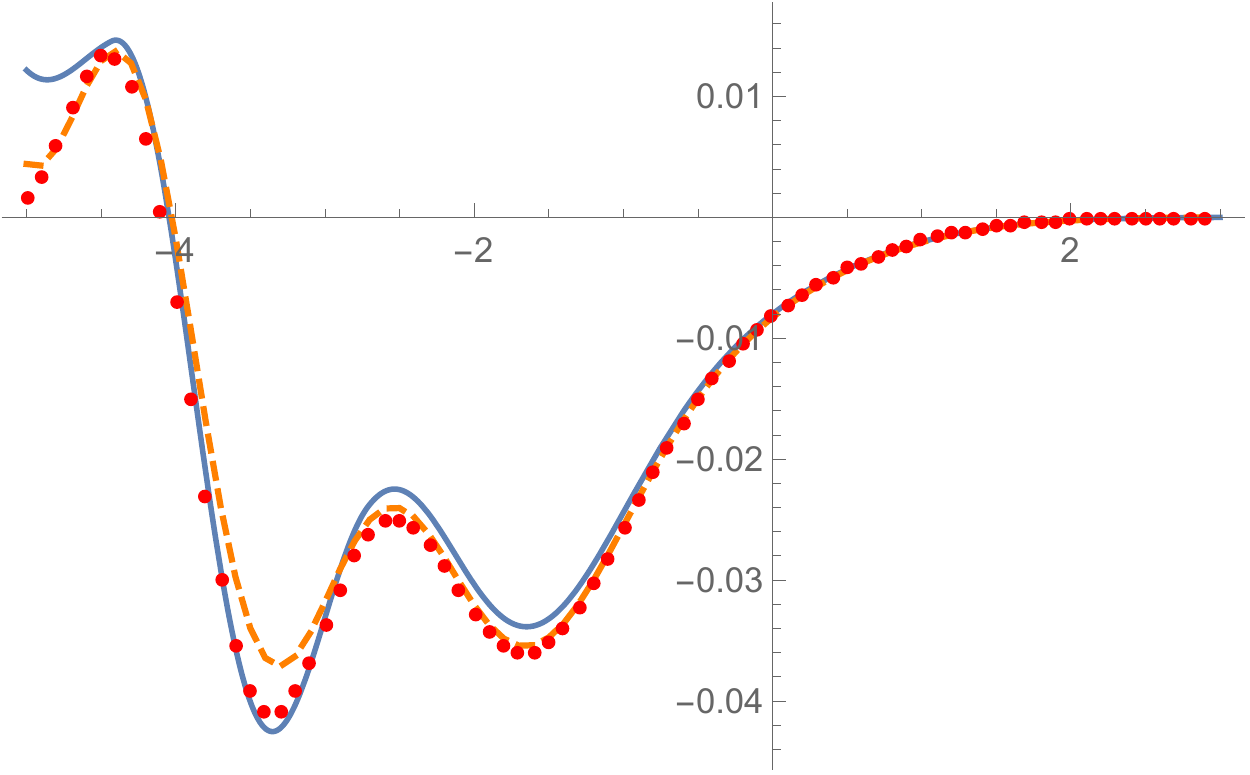}
		\caption{$\xi=0.6$}
	\end{subfigure}
	\caption{GUE. \label{fig:3} Solid line is $p_{1,\infty}^{\xi, G}(s)$ as computed using the coupled DEs, dashed line is the operator formula for $p_{1,\infty}^{\xi, G}(s)$ computed using Bornemann's method, dots are the scaled difference (\ref{6.1}).}
\end{figure}

\begin{figure}[H]
	\centering
	\begin{subfigure}[t]{0.5\textwidth}
		\centering
		\includegraphics[height=1.6in]{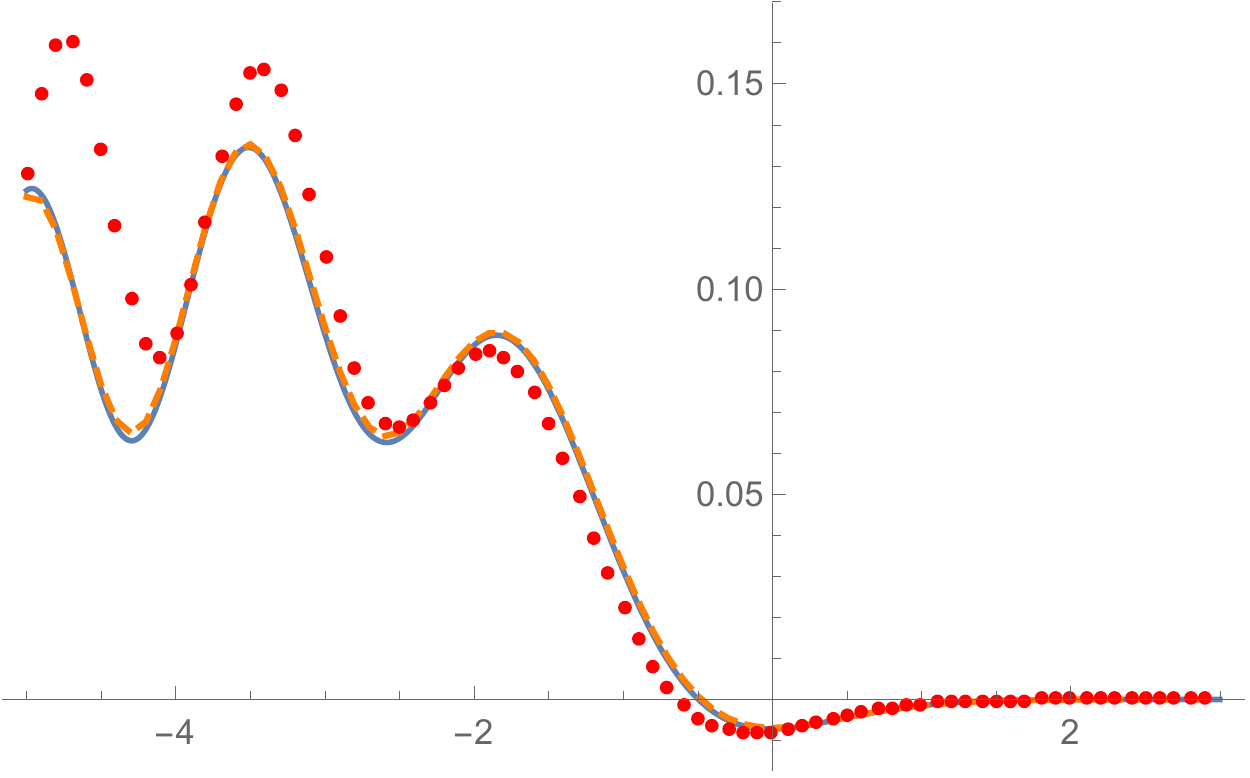}
		\caption{$\xi=0.3$}
	\end{subfigure}%
	~
	\begin{subfigure}[t]{0.5\textwidth}
		\centering
		\includegraphics[height=1.6in]{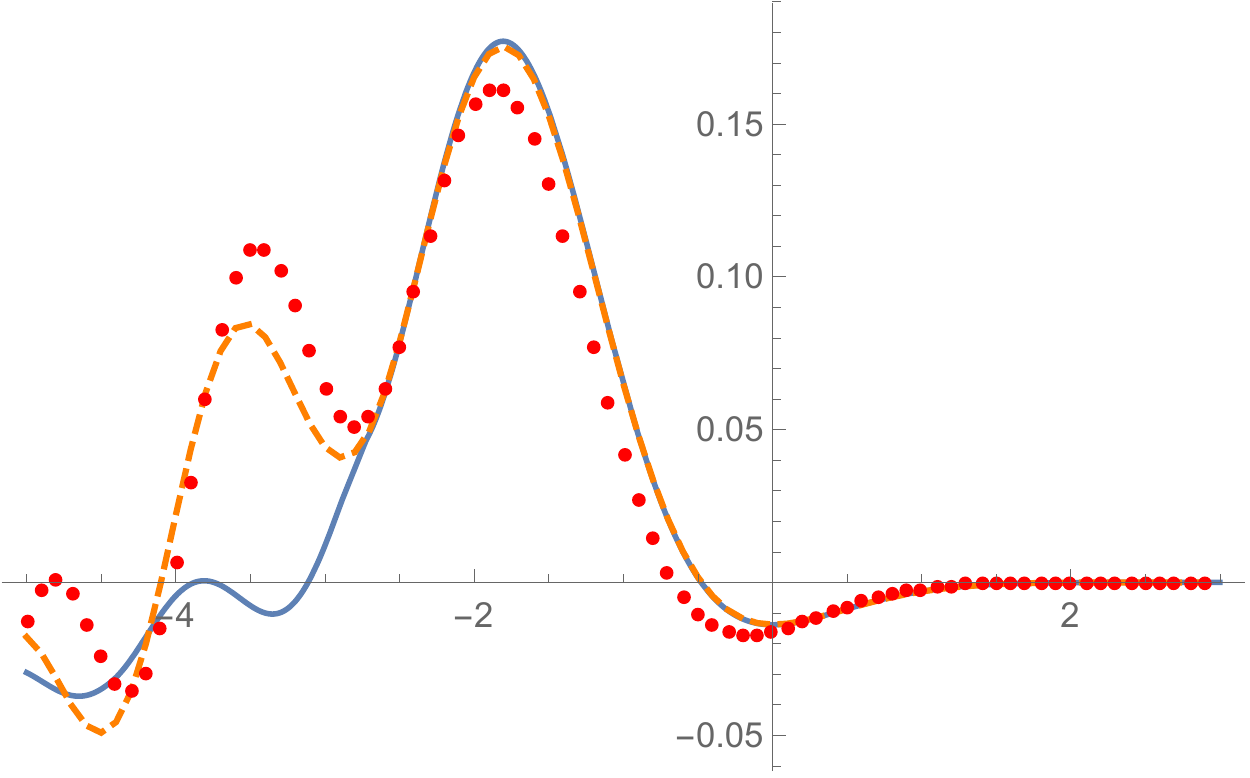}
		\caption{$\xi=0.6$}
	\end{subfigure}
	\label{fig:4}
	\caption{LUE. \label{fig:4} Solid line is $p_{1,\infty}^{\xi, L}(s)$, as computed using the coupled DEs, dashed line is the operator formula for $p_{1,\infty}^{\xi, L}(s)$ computed using Bornemann's method, dots are the scaled difference (\ref{6.1}).}
\end{figure}

\section{LUE with $a$ proportional to $N$}
\subsection{Expansion of the kernel and operator theoretic formulae}
Suppose in the LUE weight (\ref{2.2}) that $a$ is replaced by $\alpha N$, with $\alpha$ fixed.
A well defined soft edge state results by introducing the scaling variable \cite{Jo01}
\begin{equation}\label{3.25}
s_{t,\alpha} = N ( \sqrt{1 + \alpha} + 1)^2 + N^{1/3} (\sqrt{1+\alpha} + 1)
\Big ( {1 \over  \sqrt{1 + \alpha} } + 1 \Big )^{1/3} t.
\end{equation}
However, the analogue of Proposition \ref{p1} for this setting is not in the existing literature.
Our first task then is to derive such a formula.

\begin{proposition}
Let $s_{t,\alpha}$ be given by (\ref{3.25}). We have
\begin{equation}\label{2}
\Big ( {\partial s_{t,\alpha} \over \partial t} \Big ) K_N(s_{x,\al},s_{y,\al})=K(x,y)+\frac{1}{N^{2/3}}L_{\al}(x,y)
+{\rm O}\left(\frac{1}{N}\right),
\end{equation}
where $K(x,y)$ is given by (\ref{2.12}) and
\begin{multline}\label{4}
L_{\al}(x,y)=- \frac{\al^2(1+\al+\sqrt{1+\al})^{1/3}}{32(1+\al)^{5/6}(1+\sqrt{1+\al})^3}\frac{(x^2+y^2)^2\left(\mathrm{Ai}(x)\mathrm{Ai}'(y)-\mathrm{Ai}'(x)\mathrm{Ai}(y)\right)}{x-y}
\\
+\frac{1}{160(1+\al)^{2/3}(1+\sqrt{1+\al})^{2/3}(x-y)}\bigg(-8(2+\al-6\sqrt{1+\al})(x^3-y^3)\mathrm{Ai}(x)\mathrm{Ai}(y)
\\
+\Big(4(6+2\sqrt{1+\al}+3\al)(x-y)-5(\sqrt{1+\al}-1)^2(x^2+y^2)^2\Big)\mathrm{Ai}(y)\mathrm{Ai}'(x)
\\
+\Big(4(6+2\sqrt{1+\al}+3\al)(x-y)+5(\sqrt{1+\al}-1)^2(x^2+y^2)^2\Big)\mathrm{Ai}(x)\mathrm{Ai}'(y)
\\
+8(2+\al-6\sqrt{1+\al})(x^2-y^2)\mathrm{Ai}'(x)\mathrm{Ai}'(y)\bigg).
\end{multline}
\end{proposition}

\begin{proof}
We adapt the method used in \cite{GFF05} to deduce Proposition \ref{p1} in the case $x=y$.
The starting point is the scaled variant of (\ref{2.6})
\begin{equation}\label{5}
4K_N(4Nx,4Ny)=\frac{(w_N(x)w_N(y))^{1/2}}{N\lVert \pi_{N-1}\rVert^2}\frac{\pi_N(x)\pi_{N-1}(y)-\pi_{N-1}(x)\pi_N(y)}{x-y},
\end{equation}
where
\begin{equation}\label{6}
w_N(x)=x^{\al N}\exp(-4Nx)
\end{equation}
and $\pi_{N+j-1}$ are scaled monic Laguerre polynomials with the integral representation
\begin{align*}\label{7}
\pi_{N+j-1}(x)&=(-1)^{N+j-1}c_j(N)\oint\frac{\mathrm{d}z}{2\pi i}e^{-2Nzx}\frac{(z+2)^{N(1+\al)}}{z^{N+1}}\left(\frac{1}{z}+\frac{1}{2}\right)^{j-1}
\\ c_j(N)&=\frac{(N+j-1)!}{(2N)^{N+j-1}}.\numberthis
\end{align*}
The contour above is oriented positively and encircles the origin but avoids enclosing the point $z=-2$. 
Substituting in (\ref{5}) gives
the double integral expression
\begin{equation}\label{8}
4K_N(4Nx,4Ny)=\frac{c_0(N)c_1(N)}{N\lVert \pi_{N-1}\rVert^2}\frac{(w_N(x)w_N(y))^{1/2}}{x-y}J_N(x,y),
\end{equation}
where
\begin{equation}\label{9}
J_N(x,y)=\oint\frac{\mathrm{d}z_1}{2\pi i}\oint\frac{\mathrm{d}z_2}{2\pi i}e^{NS(z_1,x)+NS(z_2,y)}G(z_1,z_2)
\end{equation}
and
\begin{align*}
S(z,x)&=-2zx-\log z+(1+\al)\log(1+z/2)
\\ G(z_1,z_2)&=\left(1+\frac{z_1}{2}\right)^{-1}\left(1+\frac{z_2}{2}\right)^{-1}\left(\frac{z_1-z_2}{z_1z_2}\right).
\end{align*}
Here, the function $S(z,x)$ has two saddle points that coalesce to the same point
\begin{equation}\label{10}
z_0=-\frac{2}{\sqrt{1+\al}+1}
\end{equation}
when $x=\frac{1}{4}(\sqrt{1+\al}+1)^2$.

Defining
$$
\widehat{s}_{t,\al} =s_{t,\al}/4N, \qquad 
 b=\left( \frac{1}{\sqrt{1+\alpha}} + 1 \right) \left( \frac{\sqrt{1+\alpha}+1}{2} \right)^3,
 $$
 $$
 S(z)=S\left(z,\frac{1}{4}(\sqrt{1+\al}+1)^2\right), \qquad
 \widehat{J}_N(x,y) =J_N(\widehat{s}_{x,\al},\widehat{s}_{y,a}),
$$
and deforming the contour in \eqref{9} in the same manner as described in \cite{GFF05}
\begin{equation}\label{11}
\widehat{J}_N(x,y)\doteq \int_\mathcal{C}\frac{dz_1}{2\pi i}\exp(NS(z_1)-b^{1/3}N^{1/3}z_1x)\int_\mathcal{C}\frac{dz_2}{2\pi i}\exp(NS(z_2)-b^{1/3}N^{1/3}z_2y) G(z_1,z_2),
\end{equation}
where $\mathcal{C}$ is the contour consisting if two rays of unit length starting at $z = e^{-i\pi/3} + z_0$ then moving to $z = z_0$ and ending at $z = e^{i\pi/3} + z_0$.  The symbol $\doteq$ is used to denote that a remainder
term exponentially small in $N$ has been ignored. Since $S(z)$ is analytic on $\mathcal{C}$, it admits a power series expansion about the point $z=z_0$
\begin{equation}\label{12}
S(z)=S(z_0)+\frac{b}{3}(z-z_0)^3+\frac{(z-z_0)^3}{3}\varphi(z-z_0),
\end{equation}
where
\begin{equation}\label{13}
\varphi(t)=\sum_{k=4}^\infty\frac{S^{(k)}(z_0)}{S^{(3)}(z_0)}\frac{3!}{k!}t^{k-3}.
\end{equation}
Substituting \eqref{13} into \eqref{11} and then taking the change of variables $t=z-z_0$, $\mathcal{C}\to\mathcal{B}$ shows
\begin{multline}\label{14}
\widehat{J}_N(x,y)\doteq e^{-b^{1/3}N^{1/3}z_0(x+y)+2N \Re[S(z_0)]}
\\ \times\int_{\mathcal{B}}\frac{dt_1}{2\pi i}\exp(bN\frac{t_1^3}{3}-b^{1/3}N^{1/3}t_1x)\int_{\mathcal{B}}\frac{dt_2}{2\pi i}\exp(bN\frac{t_2^3}{3}-b^{1/3}N^{1/3}t_2y)
\\ \times G(t_1+z_0,t_2+z_0)e^{bN\frac{t_1^3}{3}\varphi(t_1)+bN\frac{t_2^3}{3}\varphi(t_2)}.
\end{multline}
If we now define the integral operator
\begin{equation*}
\mathbb{E}_{t,x}=\int_{\mathcal{B}}\frac{dt}{2\pi i}\exp(bN\frac{t^3}{3}-b^{1/3}N^{1/3}tx)
\end{equation*}
then for $b>0$ and $0<\beta<1/3$,
\begin{equation*}
\mathbb{E}_{t,x}t^m = (-1)^m b^{-(m+1)/3} N^{-(m+1)/3} [\mathrm{Ai}^{(m)}(x) + {\rm O}(e^{-\beta b N})].
\end{equation*}
This is immediate from the contour integral expression of the Airy function after the change of variables $z=(bN)^{1/3}t$. It allows us to expand the integral component \eqref{14} as
\begin{equation}\label{15}
N\mathbb{E}_{t_1,x}\mathbb{E}_{t_2,y}F(\lambda_1,\lambda_2,t_1,t_2)=c_0(x,y)+\frac{1}{N^{1/3}}c_1(x,y)+\frac{1}{N^{2/3}}c_2(x,y)+{\rm O}\left(\frac{1}{N}\right),
\end{equation}
where
\begin{equation}\label{16}
F(\lambda_1,\lambda_2,t_1,t_2)=G(t_1+z_0,t_2+z_0)e^{\lambda_1\varphi(t_1)+\lambda_2\varphi(t_2)}
\end{equation}
and the functions $N^{-m/3}c_m(x,y)$ can be computed from the formula
\begin{equation}\label{17}
c_m(x,y) = \sum_{k=0}^{m+1}\frac{1}{k!(m+1-k)!}N\mathbb{E}_{t_1,x}\mathbb{E}_{t_2,y}t_1^k t_2^{m+1-k}\frac{\partial^k}{\partial s_1^k}\frac{\partial^{m+1-k}}{\partial s_2^{m+1-k}}F(\lambda_1,\lambda_2,s_1,s_2)\bigg\rvert_{s_1,s_2=0,\lambda_i=bN\frac{t_i^3}{3}}.
\end{equation}
For the unaccounted factors,
\begin{multline}\label{18}
\frac{c_0(N)c_1(N)}{\lVert \pi_{N-1}\rVert^2}(w_N(\widehat{s}_{x,\al})w_N(\widehat{s}_{y,a}))^{1/2}e^{-b^{1/3}N^{1/3}z_0(x+y)+2N \Re[S(z_0)]}
\\ = N\frac{\sqrt{1+a}}{2}\bigg(1-\frac{\al}{4N^{1/3}}\frac{(1/\sqrt{1+\al}+1)^{2/3}}{(\sqrt{1+\al}+1)^2}(x^2+y^2)
\\ +\frac{\al^2}{32N^{2/3}}\frac{(1/\sqrt{1+\al}+1)^{4/3}}{(\sqrt{1+\al}+1)^4}(x^2+y^2)^2+{\rm O}\left(\frac{1}{N}\right)\bigg) 
\\ = N\left(g_0(x,y)+\frac{1}{N^{1/3}}g_1(x,y)+\frac{1}{N^{2/3}}g_2(x,y)+{\rm O}\left(\frac{1}{N}\right)\right).
\end{multline}
Combining the results \eqref{15}, \eqref{16}, \eqref{18} gives
\begin{multline}\label{19}
(bN)^{1/3}K_N(s_{x,\al},s_{y,\al})=\frac{1}{x-y}\bigg(g_0(x,y)c_0(x,y) + [g_0(x,y)c_1(x,y)+g_1(x,y)c_0(x,y)]\frac{1}{N^{1/3}}
\\ +[g_0(x,y)c_2(x,y)+g_1(x,y)c_1(x,y)+g_2(x,y)c_0(x,y)]\frac{1}{N^{2/3}}+{\rm O}\left(\frac{1}{N}\right)\bigg).
\end{multline}
With the assistance of computer algebra, the first term in \eqref{19} is the Airy kernel (\ref{2.12}), the second term vanishes and the third term is \eqref{4}. \end{proof}

\begin{remark} (a) Setting $\alpha = 0$, \eqref{4} reduces to  (\ref{2.13}). (b) Setting $a = \alpha N$ in (\ref{3.34})
and expanding
$$
\Big ( {\partial s_{t,\al} \over \partial t} \Big ) \rho_N(s_{y,\alpha}) =
\rho_{0,\infty}(y;\alpha) + {1 \over N^{2/3}}  \rho_{1,\infty}(y;\alpha) + \cdots
$$
shows that $ \rho_{1,\infty}(y;\alpha)$ satisfies the third order linear DE
\begin{multline}
\left(\frac{1}{\sqrt{1+\alpha}}+1\right)^{2/3}(1+\alpha)(\rho_1'''-4y\rho_1'+2\rho_1) 
\\ = -\sqrt{1+\alpha}(3y\rho_0'''+4\rho_0''-6y^2\rho_0')+(2+\alpha)(y^2\rho_0'-y\rho_0).
\label{4.19}
\end{multline}
But $\rho_{1,\infty}(y;\alpha) = L_\alpha(x,x)$ as specified in \eqref{4}. Indeed, with the aid of computer algebra,
this can be checked to satisfy (\ref{4.19}). A further point of interest is that this procedure
also exhibits that $\rho_0$ satisfies a third order linear differential equation, with the same
homogeneous part as in (\ref{4.19}).

\end{remark}

Corollary \ref{cor2} again applies, now with $\mathbb L_{(t,\infty)}$ denoting the integral operator on $(t,\infty)$ with
kernel (\ref{4}). And the operator theoretic formulae therein can be computed with Bornemann's method as detailed
in \S \ref{NE}. Some examples are given in Figures \ref{fig:5a} and \ref{fig:5b}.

\begin{figure}[H]
	\centering
	\begin{subfigure}[t]{0.5\textwidth}
		\centering
		\includegraphics[height=1.6in]{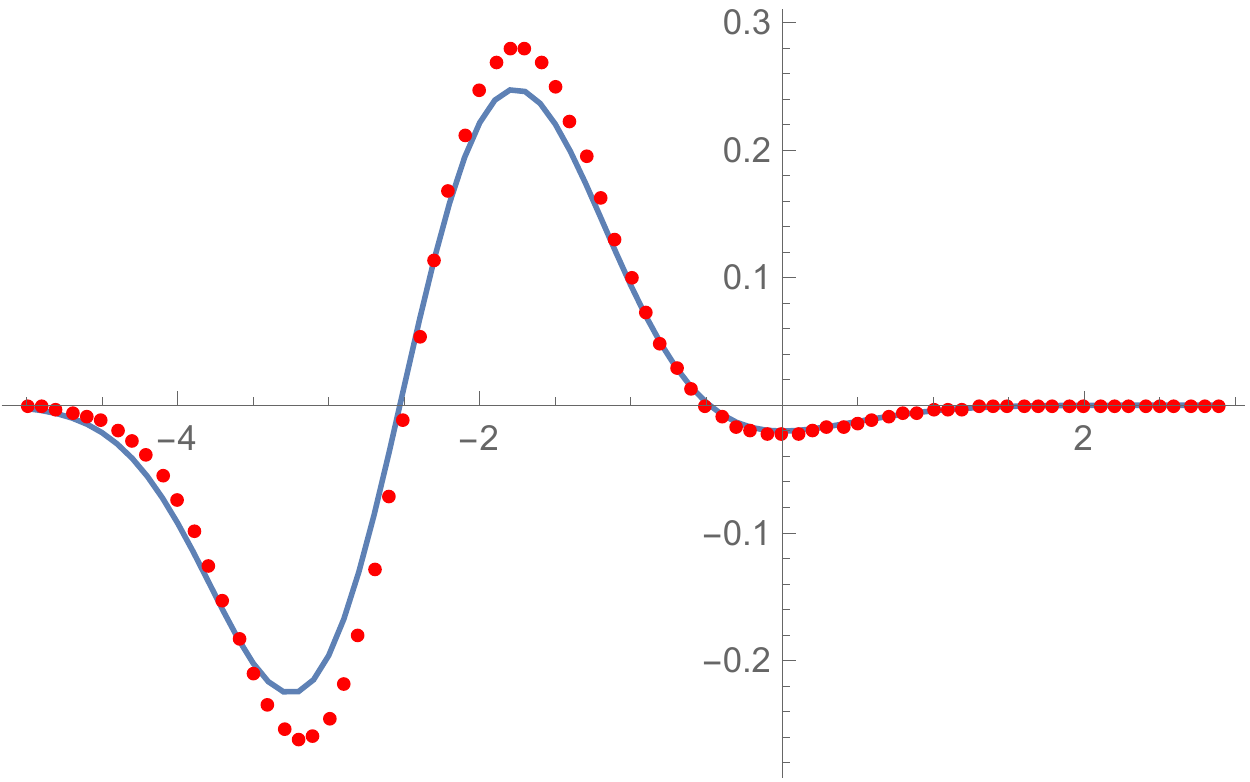}
		\caption{$\alpha=0.5$, $\xi = 1$}
		\label{fig:5a}
	\end{subfigure}%
	~
	\begin{subfigure}[t]{0.5\textwidth}
		\centering
		\includegraphics[height=1.6in]{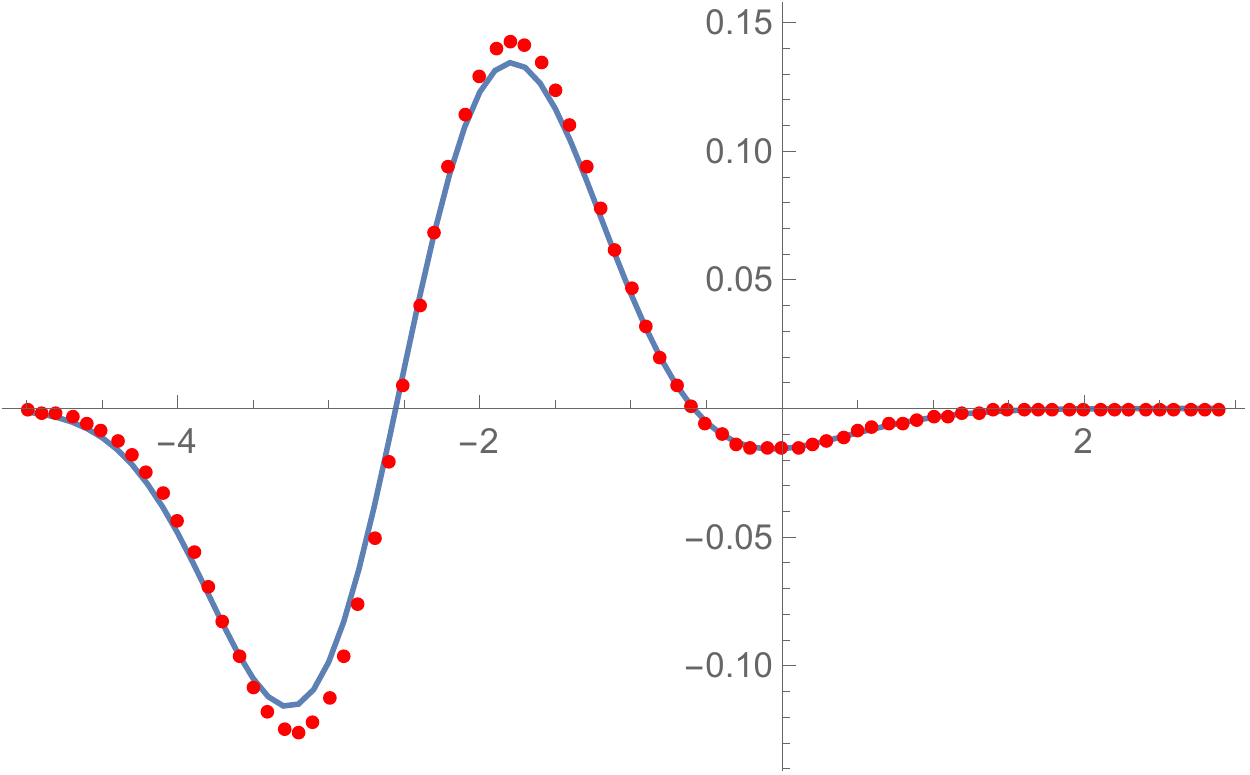}
		\caption{$\alpha=5$, $\xi = 1$}
		\label{fig:5b}
	\end{subfigure}
	\caption{[Colour-on-line] The solid (blue) curve is the correction term derived from Bornemann's method,
	while the dots (red) are the scaled difference (\ref{6.1}).}
	\label{fig:5}
\end{figure}

\subsection{Coupled differential equations}\label{s4.2}

A characterisation in terms of coupled differential equations is also possible. Starting with (\ref{3.12}),
this is obtained by expanding
$$
\Big ( {\partial s_{t,\alpha} \over \partial t} \Big )
{U^L(x_{y,\alpha}; \xi) \over x_{y,\alpha}} =
\sigma_0^{L,\alpha}(y) + {1 \over N^{2/3}} \sigma_1^{L,\alpha}(y) + {\rm O} \Big ( {1 \over N} \Big ),
$$
where $s_{t,\alpha}, x_{y,\alpha}$ are defined as is consistent with (\ref{3.25}).
Noting too that as a corollary of (\ref{2}), 
$$
\Big ( {\partial s_{t,\alpha} \over \partial t} \Big ) \rho_N^L(x_{y,\alpha}) =
\rho_{0,\infty}^L(y) + {1 \over N^{2/3}}  \rho_{1,\infty}^{L,\alpha}(y) + {\rm O} \Big ( {1 \over N} \Big ),
$$
and $ \rho_{1,\infty}^{L,\alpha}(y) = L_\alpha(x,x)$, where $\rho_{0,\infty}^L(y)$ is given by
(\ref{2.12}), we have all the information required to derive the analogue of
Propositions \ref{p3} and \ref{p4}.

\begin{proposition}\label{p6}
Consider the expansion \eqref{2.15} for the LUE with $a = \alpha N$, $\alpha$ fixed independent of $N$. Set $\sigma_0(y;\xi)=\sigma_0^G(y;\xi)$ where $\sigma_0^G(y;\xi)$ is specified by \eqref{3.6} and \eqref{3.7}. Define $\sigma_1^{L,\alpha}$ as the solution of the second order linear differential equation \eqref{3.8} with
\begin{align*}
A(y)&=2(1+\alpha)\sigma_0'' \\
B(y)&=4(1+\alpha)\left(3(\sigma_0')^2 - 2y\sigma_0' + \sigma_0\right) \\
C(y)&=4(1+\alpha)\sigma_0' \\
D(y)&=
\bigg(1+\frac{1}{\sqrt{1+\alpha}}\bigg)^{1/3}\bigg\{  \frac{\alpha^2\sqrt{1+a}}{(1+\sqrt{1+\alpha})^3}(\sigma_0)^2
-2(1+\alpha+\sqrt{1+\alpha})y\sigma_0\sigma_0'
\\ &\mathrel{\phantom{D(y)=}}
+\frac{1}{\alpha}\left(-8+8\sqrt{1+\alpha}-7\alpha+\alpha^2+9\alpha\sqrt{1+\alpha}\right)y^2(\sigma_0')^2
\\ &\mathrel{\phantom{D(y)=}}
-\frac{4(1+\alpha)}{1+\sqrt{1+a}}\left(2\sigma_0(\sigma_0')^2+\sigma_0'\sigma_0''+y(\sigma_0'')^2+4y(\sigma_0')^3\right)\bigg\}.
\end{align*}
The solution $\sigma_1^{L,\alpha}$ is to be chosen subject to the boundary condition
$$
\sigma_1^{L,\alpha}(y) \mathop{\sim}\limits_{y \to \infty} \xi L_\alpha(y,y).
$$
The final formulas of Proposition \ref{p3} remain true with all superscripts $G$ replaced by
$L,\alpha$.
\end{proposition}

In Figures \ref{fig:6a} and \ref{fig:6b} we plot some graphs obtained from solving the coupled differential
equations of Proposition \ref{p6} in the case $\xi = 1$.

\begin{figure}[H]
	\centering
	\begin{subfigure}[t]{0.5\textwidth}
		\centering
		\includegraphics[height=1.6in]{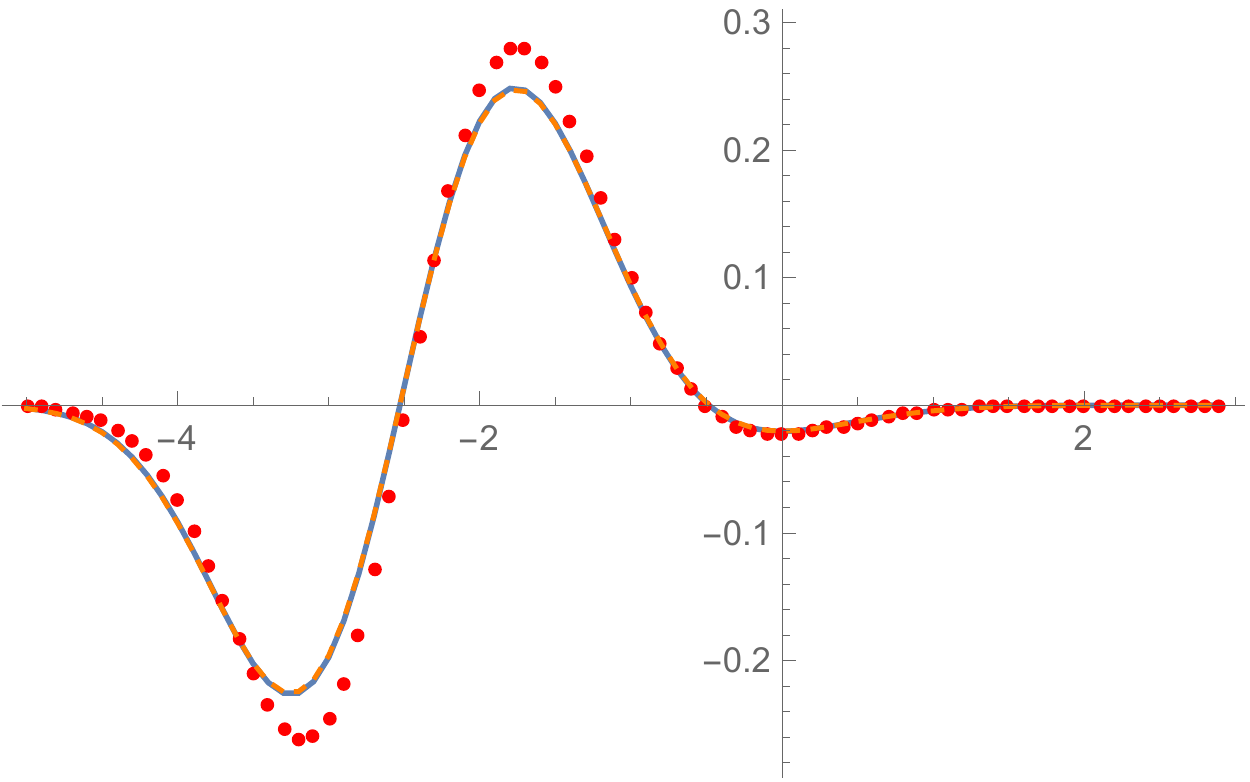}
		\caption{$\alpha=0.5$}
		\label{fig:6a}
	\end{subfigure}%
	~
	\begin{subfigure}[t]{0.5\textwidth}
		\centering
		\includegraphics[height=1.6in]{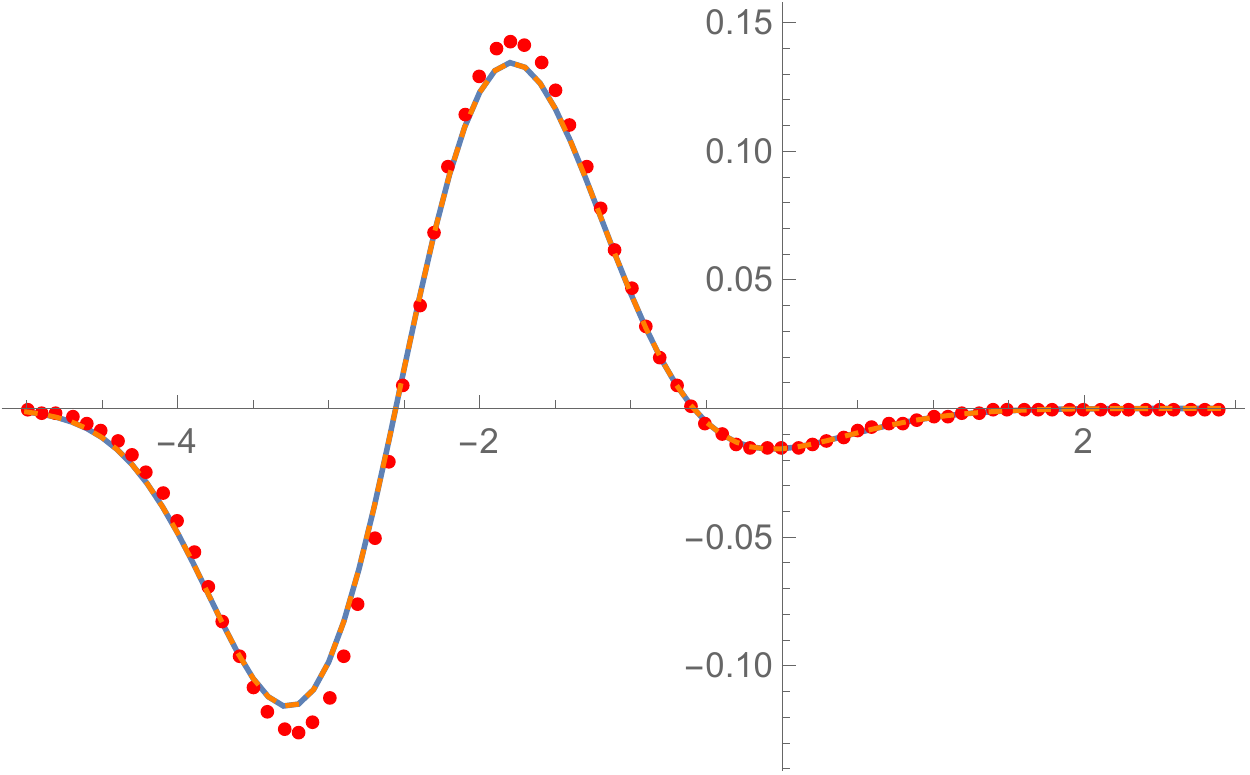}
		\caption{$\alpha=5$}
		\label{fig:6b}
	\end{subfigure}
	\caption{[Colour-on-line] The solid (blue) curve is the correction term derived from the Painlev\'e equations. The dashes (in orange) are the correction term derived from Bonrnemann's method. The dots (in red) are the scaled difference (\ref{6.1}).}
	\label{fig:6}
\end{figure}

\subsection{Consequences for a particular stochastic directed path model}\label{s4.4}
A point of interest in relation to the largest eigenvalue of the LUE is that the associated gap probability
shows itself as an exact formula for a probability in a particular stochastic model on a rectangular grid,
where the observable involves a weighted path. Thus consider an $N \times n$ rectangular grid, with
$n \ge N$ for definiteness. Associate with each lattice site $(i,j)$ an independent exponential random
variable $x_{i,j}$ of unit variance. Define
$$
\ell (N,n) = {\rm max} \sum_{{\rm up/right} \atop
(1,1) \: {\rm to} \: (N,n)} x_{i,j},
$$
where ``up/right $(1,1)$ to $(N,n)$" restricts the indices in the sum to form a path on the
grid which starts at (1,1), and finishes at $(N,n)$, going along the grid in steps which go either one
step up, or one step to the right.

It is a known theorem \cite{Jo99a,BR01}, using the notation defined above (\ref{2.7}), that
\begin{equation}\label{7.1}
{\rm Pr} \, (\ell(N,n) \le s) = E_N^{\rm LUE}(0;(s,\infty)),
\end{equation}
where the LUE has $a=n-N$. In the case that $N \to \infty$ and $a$ fixed,
the results of \S \ref{s2} and \S \ref{s3} as they relate to the LUE, and in particular the LUE
analogue of (\ref{3.11a}), imply the explicit form of the leading correction to the large $N$
form of ${\rm Pr} \, (\ell(N,n) \le s)$, as do the results of \S \ref{s4.2} for $N \to \infty$ with $n - N = \alpha N$ ($\alpha$ fixed).

\begin{corollary}
Let $\ell(N,n)$ be specified as above. For $N \to \infty$ with $n-N$ fixed,
\begin{multline}
{\rm Pr} \Big ( (\ell (N,n) - (4N + 2(n-N))/(2(2N)^{1/3}))\le s  \Big )  \nonumber \\
=\exp \Big ( - \int_s^\infty \sigma_0^{G}(t;1) \, dt \Big )
\Big ( 1 - {1 \over N^{2/3}} \int_s^\infty \sigma_1^L(t;1) \, dt + {\rm O} \Big ( {1 \over N} \Big ) \Big ),
\end{multline}
while for $N \to \infty$ with $n - N = \alpha N$ ($\alpha$ fixed)
\begin{multline}
{\rm Pr} \Big ( (\ell (N,n) - N(\sqrt{1 + \alpha} + 1)^2 ))/( c_\alpha N^{1/3}) \le s  \Big ) \nonumber \\
=
\exp \Big ( - \int_s^\infty \sigma_0^{G}(t;1) \, dt \Big )
\Big ( 1 - {1 \over N^{2/3}} \int_s^\infty \sigma_1^{L,\alpha}(t;1) \, dt + {\rm O} \Big ( {1 \over N} \Big ) \Big ),
\end{multline}
where
$$
c_\alpha = (\sqrt{1 + \alpha} + 1) \Big ( {1 \over \sqrt{1 + \alpha} + 1} + 1 \Big )^{1/3}.
$$
\end{corollary}

\section{Numerical study of the distribution of $\lambda_{\rm max}$ for a particular Wigner
matrix}

The exact asymptotic analysis presented above shows that for both the GUE and LUE models, and
upon appropriate centring and scaling, the PDF for $\lambda_{\rm max}$ has the large $N$ form
(\ref{2.15}). In this expansion, the observed universality (i.e.~model independence) of the leading term
$p_{0,\infty}^{\xi = 1}(t)$ is well established in the case of Wigner matrices --- see  see e.g.~\cite{Er11}
--- and sample covariance (Wishart) matrices --- see e.g.~\cite{Pe09}. Wigner matrices are random Hermitian matrices with all entries above the diagonal identically and independently distributed
with mean zero. The GUE is an example of a Wigner matrix with Gaussian entries. Sample covariance matrices are matrices of the form $X^\dagger X$, with the entries of $X$ independent and identically distributed with mean zero.
Such matrices in the case of Gaussian entries give a construction of the LUE; see e.g.~\cite[Ch.~3]{Fo10}.

Our study has shown that the functional form of the leading correction $p_{1,\infty}^\xi(t)$ is model
dependent. On the other hand, this leading correction has the same $N$ dependence in all cases
considered, being proportional to $N^{-2/3}$. One may then ask if this dependence on $N$ persists for
a wider class of Wigner or Wishart matrices? Certainly the strategy of the present
study is no longer feasible, as beyond the Gaussian case there are no examples of Wigner or Wishart matrices for which exact finite $N$ formulas are available for asymptotic analysis. To initiate a study in
this direction, we use instead numerical simulation for one class of Wigner ensemble.

The particular Wigner ensemble to be considered is constructed as $Y = {1 \over 2} (X + X^\dagger)$, where
all entries of $X$ are chosen independently and uniformly from the set of four values
${1 \over \sqrt{2}} (\pm 1 \pm i )$. From this construction the off diagonal entries of $Y$ have mean
zero and variance one half, as for the GUE. It is known (see e.g.~\cite{PS11}) that in this circumstance
the largest eigenvalue is equal to $\sqrt{2N}$, and moreover from \cite{Er11} that with the GUE
scaling as in (\ref{2.8}), the PDF of $\lambda_{\rm max}$ has the universal limiting form 
$p_{0,\infty}^{\xi = 1}(t)$. As our first investigation, for various fixed values of $N$ we formed a histogram
of the distribution of the random variable 
\begin{equation}\label{cc}
t = \sqrt{2} N^{1/6}(\lambda_{\rm max} - \sqrt{2N}), 
\end{equation}
and subtracted from the histogram $p_{0,\infty}^{\xi = 1}(t)$. It was observed that multiplying this difference by
$N^{1/3}$ gave a functional form which to leading order appeared to be independent of $N$; see Figure \ref{fig:7a}.

\begin{figure}[t]
	\centering
	\begin{subfigure}[t]{0.5\textwidth}
		\centering
		\includegraphics[height=2.0in]{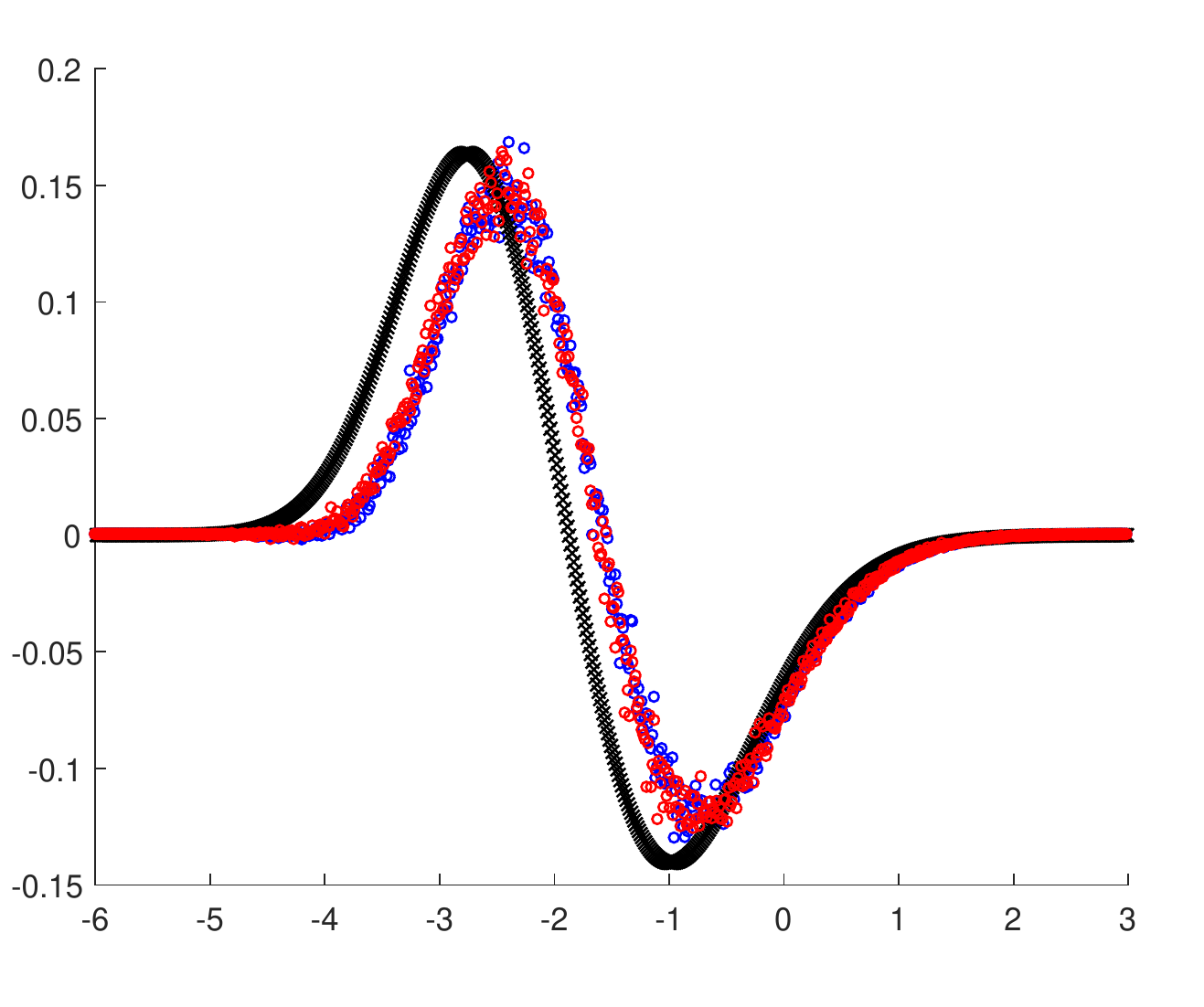}
		\caption{Scaled variable (\ref{cc})}
		\label{fig:7a}
	\end{subfigure}%
	~
	\begin{subfigure}[t]{0.5\textwidth}
		\centering
		\includegraphics[height=2.0in]{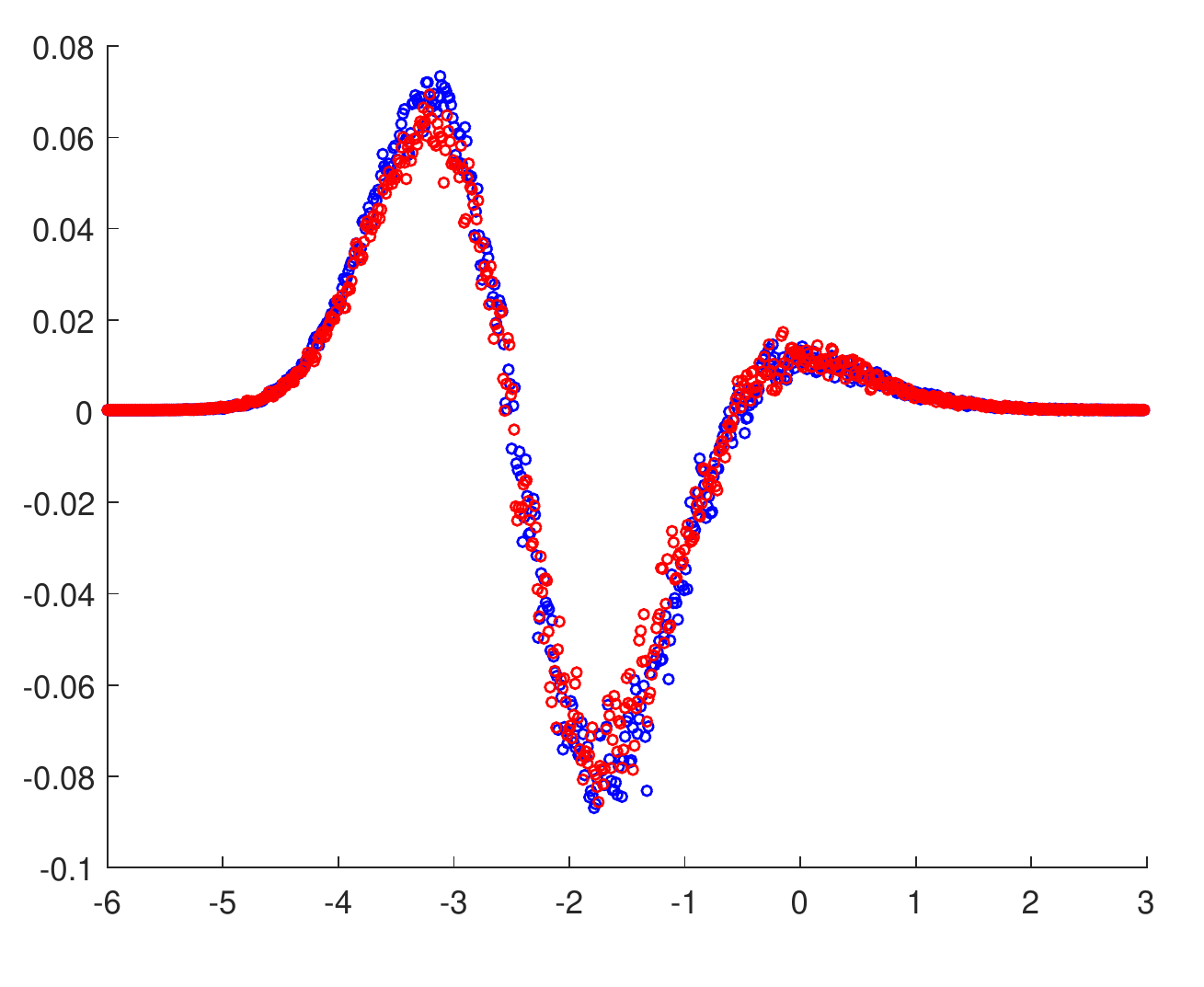}
		\caption{Scaled variable (\ref{cc2})}
		\label{fig:7b}
	\end{subfigure}
	\caption{[Colour-on-line] Blue points $N=50$; red points $N=60$. The left figure is the difference between the histogram for the scaled variable $t$ in (\ref{cc}) and 
	$p_{0,\infty}^{\xi=1}(t)$, multiplied by $N^{1/3}$. Also shown plotted (in black) is one half of 
	${d \over d t}
p_{0,\infty}^{\xi=1}(t)$.The right figure is the difference between the histogram for the scaled variable $t^*$ in (\ref{cc2}) and 
	$p_{0,\infty}^{\xi=1}(t)$, multiplied by $N^{2/3}$. The histograms were each formed from $5 \times 10^6$ samples.}
\end{figure}

So at this stage the particular Wigner ensemble under consideration is exhibiting a large $N$
form for the distribution of the centring and scaling of $\lambda_{\rm max}$ (\ref{cc}) which has
the first correction term  proportional to $N^{-1/3}$ rather than
$N^{-2/3}$ as in (\ref{2.15}). However, we still have the freedom to alter the precise choice of
centring and scaling. To see the possible effect, note that in (\ref{2.15}), with $c$ a constant,
\begin{equation}\label{cc1}
p_N^\xi(t + c/N^{1/3} )=p_{0,\infty}^\xi(t)+\frac{c}{N^{1/3}}{d \over d t}
p_{0,\infty}^\xi(t)+{\rm O}\left(\frac{1}{N^{2/3}}\right).
\end{equation}
Thus one possible mechanism for the leading correction to be proportional to $N^{-1/3}$ is that
the centring and scaling has not been chosen optimally. Furthermore, the signature of this is that
the functional form of the leading correction is then proportional to the derivative of the leading,
universal form.

Given these considerations, we then compared our graphs with ${d \over d t}
p_{0,\infty}^{\xi=1}(t)$ and indeed found a close resemblance, provided we multiplied the latter
by one half; see Figure \ref{fig:7a}. This suggested that (\ref{cc1}) holds with $c=1/2$, and that
introducing the scaled variable
\begin{equation}\label{cc2}
t^* =  \sqrt{2} N^{1/6} (\lambda_{\rm max} - \sqrt{2N} )  + 1/ ( 2 N^{1/3}), 
\end{equation}
does indeed give a correction  which is proportional to 
$N^{-2/3}$. Evidence that this is indeed the case is given in Figure \ref{fig:7b}.

\section*{Acknowledgements}
This research project is part of the program of study supported by the 
ARC Centre of Excellence for Mathematical \& Statistical Frontiers.


\providecommand{\bysame}{\leavevmode\hbox to3em{\hrulefill}\thinspace}
\providecommand{\MR}{\relax\ifhmode\unskip\space\fi MR }
\providecommand{\MRhref}[2]{%
  \href{http://www.ams.org/mathscinet-getitem?mr=#1}{#2}
}
\providecommand{\href}[2]{#2}

\end{document}